\pdfoutput=1
\documentclass[11pt,reqno]{article}
\usepackage{lmodern}          %
\usepackage[T1]{fontenc}
\usepackage{amsmath,amssymb,amsthm}
\usepackage{tikz-cd}
\usetikzlibrary{arrows.meta}
\usepackage[margin=1in]{geometry}
\usepackage{microtype}
\usepackage{enumitem}
\usepackage{booktabs}          %
\usepackage{array}             %
\usepackage{float}
\usepackage{xcolor}            %
\usepackage{xurl}
\usepackage[colorlinks=true,linkcolor=blue!55!black,citecolor=green!45!black,%
            urlcolor=blue!55!black]{hyperref}
\usepackage{iftex}
\ifpdf
  \usepackage{embedfile}
\fi
\hypersetup{
  pdftitle={Coupling Independence Implies Zero-Freeness},
  pdfauthor={Shuai Shao and Ke Shi},
  pdfkeywords={Potts model, Fisher zeros, Lee\textendash Yang zeros,
    coupling independence, zero-freeness, Holant problems}
}
\usepackage{aliascnt}
\usepackage{cleveref}

\theoremstyle{plain}
\newtheorem{theorem}{Theorem}[section]
\IfFormatAtLeastTF{2026-06-01}{%
  \newtheorem{proposition}[theorem]{Proposition}
  \newtheorem{lemma}[theorem]{Lemma}
  \newtheorem{corollary}[theorem]{Corollary}
  \theoremstyle{definition}
  \newtheorem{definition}[theorem]{Definition}
  \theoremstyle{remark}
  \newtheorem{remark}[theorem]{Remark}
}{%
  \newaliascnt{proposition}{theorem}
  \newtheorem{proposition}[proposition]{Proposition}
  \aliascntresetthe{proposition}
  \newaliascnt{lemma}{theorem}
  \newtheorem{lemma}[lemma]{Lemma}
  \aliascntresetthe{lemma}
  \newaliascnt{corollary}{theorem}
  \newtheorem{corollary}[corollary]{Corollary}
  \aliascntresetthe{corollary}
  \theoremstyle{definition}
  \newaliascnt{definition}{theorem}
  \newtheorem{definition}[definition]{Definition}
  \aliascntresetthe{definition}
  \theoremstyle{remark}
  \newaliascnt{remark}{theorem}
  \newtheorem{remark}[remark]{Remark}
  \aliascntresetthe{remark}
}

\crefname{theorem}{Theorem}{Theorems}
\crefname{proposition}{Proposition}{Propositions}
\crefname{lemma}{Lemma}{Lemmas}
\crefname{corollary}{Corollary}{Corollaries}
\crefname{definition}{Definition}{Definitions}
\crefname{remark}{Remark}{Remarks}
\crefname{section}{Section}{Sections}
\crefname{equation}{}{}

\newcommand{\C}{\mathbb{C}}
\newcommand{\eps}{\varepsilon}
\newcommand{\abs}[1]{\left|#1\right|}

\newcommand{\set}[1]{\left\{#1\right\}}
\newcommand{\dist}{\operatorname{dist}}
\newcommand{\Log}{\operatorname{Log}}          %
\newcommand{\one}[1]{\mathbf{1}_{\left[#1\right]}}   %
\DeclareMathOperator{\E}{\mathbb{E}}
\DeclareMathOperator{\Ham}{Ham}

\newcommand{\Deg}{\Delta}                        %
\newcommand{\Zpart}{Z}                           %

\newcommand{\Wd}{W_{1,d}}                        %
\newcommand{\WHam}{W_{1,\Ham}}                   %

\newcommand{\colours}{[q]}                        %
\newcommand{\bcount}{b}
\newcommand{\Gdeg}{\mathcal{G}_{\Deg}}
\newcommand{\Gdegg}[1]{\mathcal{G}_{\Deg,#1}}
\newcommand{\Zpin}[2]{Z_{#1}^{#2}}
\newcommand{\nZpin}[2]{\widehat Z_{#1}^{#2}}
\newcommand{\mupin}[3]{\mu_{#1,#3}^{#2}}
\newcommand{\hardlist}[2]{L^{#1}(#2)}
\newcommand{\gapq}{\gamma_0}

\newcommand{\flip}{\Phi}                           %
\newcommand{\softK}{K}

\title{\bfseries Coupling Independence Implies Zero-Freeness}
\author{Shuai Shao\thanks{School of Computer Science and Technology,
    University of Science and Technology of China.}\\
  \href{mailto:shao10@ustc.edu.cn}{\texttt{shao10@ustc.edu.cn}}
\and
Ke Shi\footnotemark[1]\\
\href{mailto:self.ke.shi@gmail.com}{\texttt{self.ke.shi@gmail.com}}}
\date{}

\begin{document}
\maketitle

\begin{abstract}
For \(\Delta\ge2\) and \(q\ge11\Delta/6\), we prove that the antiferromagnetic
\(q\)-state Potts partition function on finite simple graphs of maximum
degree at most \(\Delta\) has no Fisher zeros in a graph-uniform complex
neighbourhood of \([0,1]\).  For \(q>11\Delta/6\), the proof establishes
coupling independence throughout \([0,1]\) using a soft version
of Vigoda's flip dynamics.  Our main tool is a separator-shell transfer
theorem for the Potts model on induced-subgraph closed classes of graphs
of maximum degree at most \(\Delta\), with \(q\ge\Delta+1\).  It yields a
graph-uniform zero-free neighbourhood of \([0,1]\) from Hamming coupling
independence at \(0\) and a uniform coupling-independence bound on each
interval \([\delta,1]\), \(\delta\in(0,1]\).  We also obtain zero-free Lee--Yang
polydiscs around the uniform field for vertex- and edge-colour fields.
For Boolean Holant problems on bounded-degree graphs whose signatures
come from a fixed finite family of log-concave symmetric signatures \(f\)
with \(f(0)>0\), such as \(b\)-matchings, we obtain graph-uniform zero-free
polytubes around every bounded box of nonnegative activities; their
union is an open zero-free neighbourhood of the nonnegative orthant.
An appendix summarizes further coupling-independence
inputs and the zero-free regions they yield.
\end{abstract}

\paragraph*{AI disclosure.}
The main text was developed with assistance from GPT-5.5, particularly
in the proof of coupling independence.  The exposition was polished
with Claude Opus 5.5.  The authors assume responsibility for all claims.

\setcounter{tocdepth}{2}
\tableofcontents

\section{Introduction}
\label{sec:intro}

The Potts model, which arises in statistical physics, describes
interactions between neighbouring spins.  In this paper, we study the
antiferromagnetic Potts model and the complex zeros of its partition
function.  The model is defined on a finite simple graph \(G=(V,E)\),
with each vertex assigned one of \(q\) spins, also called colours, from
\(\colours=\{1,\ldots,q\}\).  A configuration is a mapping
\(\sigma\colon V\to\colours\).  Its weight is \(z^{m_G(\sigma)}\), where
\(z\) is an edge activity and \(m_G(\sigma)\) is the number of edges
whose endpoints receive the same colour.  The \emph{partition function}
is
\begin{equation}
  \Zpart_G(q;z):=\sum_{\sigma\colon V\to\colours}z^{m_G(\sigma)},
  \qquad
  m_G(\sigma):=\bigl|\{uv\in E:\sigma_u=\sigma_v\}\bigr|.
  \label{eq:intro-potts-Z}
\end{equation}
The model is \emph{antiferromagnetic} if the edge activity is a real
number \(x\in[0,1]\); we reserve \(z\) for complex activities.  In this
regime, each monochromatic edge contributes a factor \(x\le1\), whereas
an edge whose endpoints receive different colours contributes~\(1\).
In particular, if \(x=0\), then only configurations with no
monochromatic edges contribute nonzero weight, and \(\Zpart_G(q;0)\)
counts the \emph{proper \(q\)-colourings} of \(G\).  At \(x=1\), the
spins are independent, and \(\Zpart_G(q;1)=q^{|V|}\).

In statistical physics, when \(x\ge0\), the partition function is the
normalizing factor of the Gibbs distribution, under which a
configuration \(\sigma\) has probability
\(x^{m_G(\sigma)}/\Zpart_G(q;x)\).  The partition function also
determines the free energy, whose nonanalytic behaviour in the
thermodynamic limit is used to describe phase transitions.  Studying the
partition function at complex activities provides a way to understand
its analyticity.  The zeros of \(\Zpart_G(q;z)\) in the complex edge
activity \(z\) are called \emph{Fisher zeros}~\cite{Fisher1965}.
Throughout, \(q\) is a fixed positive integer.

From the computational perspective, approximating these partition
functions is a basic problem that has driven the development of sampling
algorithms and deterministic approximate counting.  Graph colourings are
a particularly important test case: how many colours, relative to the
maximum degree \(\Deg\), suffice for efficient approximation?

\paragraph{Approximate counting of colourings.}
The classical approach uses Markov chain Monte Carlo (MCMC): a rapidly
mixing chain supplies approximate samples, which can be converted into a
fully polynomial-time randomized approximation scheme (FPRAS) for
counting.  Jerrum~\cite{Jerrum1995} established rapid mixing of Glauber
dynamics when \(q>2\Deg\).  Vigoda~\cite{Vigoda2000} introduced a coupling
of bounded two-colour flips and improved the condition to
\(q>11\Deg/6\), a threshold that remained the benchmark for general
graphs for nearly two decades.  Subsequent refinements reached
\(q\ge(11/6-\varepsilon_0)\Deg\) for an absolute constant
\(\varepsilon_0\approx10^{-5}\)~\cite{CDMPP2019},
and \(q\ge1.809\Deg\) when \(\Deg\ge125\)~\cite{CarlsonVigoda2024}.

Deterministic algorithms developed along two complementary directions.
The first is \emph{correlation decay}, exemplified by Weitz's algorithm
for the hard-core model~\cite{Weitz2006}.  Its central idea is to
approximate local marginals by truncating a recursive computation once
the influence of distant boundary conditions becomes sufficiently small.
For colourings, Gamarnik and Katz~\cite{GamarnikKatz2012} developed a
computation-tree approach and obtained an FPTAS on triangle-free graphs.
Lu and Yin~\cite{LuYin2013} subsequently obtained an FPTAS on general
graphs when \(q\ge\alpha\Deg+1\), for every fixed
\(\alpha>\alpha_{\mathrm{LY}}\approx2.58071\).
Here and below, deterministic approximation results are understood with
\(q\) and \(\Deg\) fixed.

The second direction uses \emph{complex zero-freeness}.  Barvinok's
interpolation method approximates the logarithm of a partition function
inside a zero-free region~\cite{Barvinok2016}; the bounded-degree
coefficient-computation method of Patel and Regts~\cite{PatelRegts2017}
turns this approach into an FPTAS in the settings considered here.
For the antiferromagnetic Potts model, Bencs, Davies, Patel, and
Regts~\cite{BencsDaviesPatelRegts2021} established a zero-free complex
neighbourhood of \([0,1)\) for \(q\ge e\Deg+1\), allowing independent
edge activities.  Liu, Sinclair, and
Srivastava~\cite{LiuSinclairSrivastava2025} obtained a zero-free
neighbourhood of \([0,1]\) for \(\Deg\ge3\) and \(q\ge2\Deg\), and
hence an FPTAS for counting colourings at this threshold.
Bencs, Berrekkal, and Regts~\cite{BencsBerrekkalRegts2026} further
improved the condition to \(q\ge1.998\Deg\), for \(q,\Deg\ge3\).

More recently, Chen, Feng, Guo, Zhang, and Zou~\cite{CFFGZZ2024}
introduced a deterministic counting algorithm based on \emph{coupling
independence}.  Informally, this property says that changing a single
pinned spin admits a coupling of the resulting conditional distributions
with a bounded expected number of disagreements, uniformly over the
size of the graph.  Together with a lower bound on nonzero conditional
marginals, it yields an FPTAS and recovers the improved colouring
thresholds of the sampling algorithms above.  Thus deterministic
approximate counting now reaches, and even surpasses, the classical
\(11\Deg/6\) threshold.  However, the corresponding zero-freeness
results had not reached this threshold on general bounded-degree
graphs.  This leaves a structural question: can the coupling estimates
that support efficient sampling and deterministic counting also exclude
complex zeros?

\subsection{Our main result}

We establish a zero-free neighbourhood of the entire antiferromagnetic
interval \([0,1]\) whenever \(q\ge11\Deg/6\), including equality.
The neighbourhood is uniform over all graphs of maximum degree at most
\(\Deg\), and the result remains valid under arbitrary pinnings after
removing the zeros forced by edges between pinned vertices.

To state this stronger conditional form, let \(\Lambda\subseteq V\),
let \(\tau\colon\Lambda\to\colours\) be a pinning, and let
\(m_G(\tau)\) denote the number of monochromatic edges in \(G[\Lambda]\).
Define the \emph{normalized pinned partition function} by
\begin{equation}
  \nZpin{G}{\tau}(q;z)
  :=\sum_{\substack{\sigma\colon V\to\colours\\
                    \sigma|_\Lambda=\tau}}
       z^{m_G(\sigma)-m_G(\tau)}.
  \label{eq:intro-normalized-pinned}
\end{equation}
The usual pinned partition function consequently satisfies
\begin{equation}
  \Zpin{G}{\tau}(q;z)
  =z^{m_G(\tau)}\nZpin{G}{\tau}(q;z).
  \label{eq:intro-normalization}
\end{equation}
This normalization removes the contribution of edges whose endpoints
are both pinned.  For \(I\subseteq\C\), write
\(\mathcal U_\eps(I):=\{z\in\C:\dist(z,I)<\eps\}\).

\begin{theorem}[Potts zero-freeness]
\label{thm:intro-main}
Let \(q,\Deg\) be integers with \(\Deg\ge2\) and \(q\ge11\Deg/6\).
There exists \(\eps=\eps(q,\Deg)>0\) such that, for every finite simple
graph \(G\) of maximum degree at most \(\Deg\) and every pinning \(\tau\),
\[
  \nZpin{G}{\tau}(q;z)\ne0
  \qquad\text{for every }z\in\mathcal U_\eps([0,1]).
\]
In particular, \(\Zpart_G(q;z)\ne0\) throughout this neighbourhood.
The same holds for the usual pinned partition function whenever
\(\tau\) is proper on \(G[\Lambda]\); otherwise, its only zero in this
neighbourhood is at \(z=0\), with multiplicity \(m_G(\tau)\).
\end{theorem}

The theorem improves the previous general-graph Fisher zero-freeness
condition \(q\ge1.998\Deg\) to \(q\ge11\Deg/6\approx1.833\Deg\).
For fixed \((q,\Deg)\), it also gives an FPTAS for
\(\Zpart_G(q;x)\) at rational activities \(x\in[0,1]\), by applying
Barvinok's interpolation method~\cite{Barvinok2016} from the explicitly
computable value \(\Zpart_G(q;1)=q^{|V|}\) and computing the required
coefficients with the algorithm of Patel and Regts~\cite{PatelRegts2017}.
The neighbourhood can be chosen effectively from the explicit coupling
bounds and the finite collection of bounded local systems in the proof.
The contribution is the exclusion of complex zeros throughout this
interval, including the hard-colouring endpoint; stronger counting
thresholds at that endpoint already follow from~\cite{CFFGZZ2024}.

\subsection{From coupling independence to zero-freeness}
\label{sec:intro-transfer}

Our proof establishes a transfer principle from real couplings to
complex zero-freeness.  We say that a model satisfies \emph{Hamming
coupling independence} with constant \(C\) if, under every pinning,
fixing a free vertex to two different feasible spins produces two
conditional Gibbs distributions that admit a coupling with expected
Hamming disagreement at most \(C\) on the remaining free vertices.
The constant is independent of the graph size.
This controls the total effect of a local change and can be established
using contractive couplings of Markov chains.

The transfer uses a distance shell to separate a root from the rest of
the graph.  If a coupling has expected total disagreement at most
\(C\), then among \(L\) disjoint shells there is one on which its
expected disagreement is at most \(C/L\).  This averaging observation
also underlies the influence-decay argument
of~\cite{CFFGZZ2024}.  On bounded-degree graphs, choosing \(L\) to be
a sufficiently large constant gives a separator with a bounded interior
and a small coupling cost on its boundary.

Conditioning on the separator expresses the partition-function response
to a complex perturbation as an average under a real boundary
distribution.  Our complex-average stability lemma controls how this
average changes when the boundary distribution changes, using its
Hamming Wasserstein distance and the variation of the exterior response.
An induction bounds that variation, while the bounded interior
contributes a uniformly small error.  The small coupling cost makes the
induction close and prevents cancellation in the partition function.
At \(x=0\), configurations forbidden by the hard constraint can acquire
nonzero weight under perturbation.  We control their contribution as an
additive error, obtaining a zero-free disc that joins the
positive-activity neighbourhoods.

The following statement summarizes the transfer principle for the Potts
model.  Its formulation separates the hard-colouring endpoint from
positive activities, so it does not require a coupling bound that is
uniform as \(x\downarrow0\).

\begin{theorem}[Coupling independence implies zero-freeness, informal]
\label{thm:potts-transfer-informal}
Fix \(\Deg\ge2\) and \(q\ge\Deg+1\), and let \(\mathcal G\) be a class
of finite simple graphs of maximum degree at most \(\Deg\), closed
under taking induced subgraphs.  Suppose the Potts model on this class
satisfies coupling independence at \(x=0\), and uniformly on
\([\delta,1]\) for every fixed \(\delta>0\), under arbitrary pinnings.
Then the normalized pinned partition functions have a common zero-free
complex neighbourhood of \([0,1]\), independent of the graph and the
pinning.
\end{theorem}

The coupling assumptions at \(x=0\) are interpreted after deleting
edges between pinned vertices, consistently with
\eqref{eq:intro-normalized-pinned}.
The transfer requires no uniform lower bound on nonzero conditional
marginals.  This matters near \(x=0\): a colour forbidden by a pinned
neighbour can have positive but vanishing marginal as \(x\downarrow0\).
The precise statement is \Cref{thm:potts-transfer}, proved in
\Cref{sec:transfer-principle}.

To obtain \Cref{thm:intro-main}, we extend Vigoda's flip coupling from
proper colourings to positive Potts activities.  The identity
\[
  x\one{c=c'}+\one{c\ne c'}
  =x+(1-x)\one{c\ne c'}
\]
represents a soft interaction as a mixture of an unconstrained edge and
a hard inequality.  Conditioning on the resulting random constraints
produces a list-colouring instance.  We couple both the constraints and
the bounded flips, allowing the two instances to differ in edges and
lists near the initial disagreement.  Averaging the contraction estimate
then yields coupling independence.  For \(q>11\Deg/6\), the bound is
uniform on \([0,1]\).  At equality, it is uniform on each
\([\delta,1]\), and we use the hard-colouring coupling bound
of~\cite[Theorem~20]{CFFGZZ2024} at \(x=0\).

\paragraph{Further consequences.}
The separator argument applies more broadly to local models for which
the required coupling and bounded-interior estimates hold: its
complex-average stability lemma is a statement about arbitrary laws on
a finite metric space.  We develop the following applications.
\begin{itemize}
  \item \emph{Colour fields.}
  For proper vertex and edge colourings, the method gives Lee--Yang
  zero-freeness under independent colour fields, using the
  coupling-independence inputs specified in \Cref{sec:lee-yang}.
  \item \emph{Log-concave Holant problems.}
  For Boolean Holant problems on bounded-degree graphs with signatures
  from a fixed finite family of log-concave symmetric signatures \(f\)
  satisfying \(f(0)>0\), we obtain graph-uniform zero-free polytubes
  around every bounded box of nonnegative activities
  (\Cref{sec:holant-consequence}).
  \item \emph{Further Potts regimes.}
  Additional coupling estimates, primarily from the companion
  article~\cite{ShaoShi2026FurtherPotts}, give further zero-free
  regions through the same transfer principle.  These include the
  condition \(q\ge1.809\Deg\) for \(\Deg\ge125\); the inputs and
  consequences are collected in \Cref{app:additional-potts}.
\end{itemize}

\subsection{Related work}

\paragraph{Correlation decay and complex zeros.}
The connection between spatial mixing and analyticity has a long history
in statistical physics.  For lattice systems, Dobrushin and Shlosman
related suitable strong mixing conditions to \emph{complete
analyticity}, which includes nonvanishing of finite-volume partition
functions under small complex
perturbations~\cite{DobrushinShlosman1985,DobrushinShlosman1987}.
On bounded-degree graphs, a substantial line of work extends contraction
estimates for real ratio recursions to complex parameters, proving
zero-freeness for the hard-core model, the Ising model, and more general
spin systems~\cite{PetersRegts2019,LiuSinclairSrivastava2019,PetersRegts2020,ShaoSun2021,LiuSinclairSrivastava2025}.
Recent work treats the hard-core model on graphs of bounded connective
constant~\cite{ChenShaoShi2026}.
Peters, Regts, and Reppekus~\cite{PetersRegtsReppekus2026} show that,
for suitable hard-core graph families, a strengthening of strong spatial
mixing called \emph{very strong spatial mixing} implies zero-freeness.

For the Potts model, the proofs of Liu, Sinclair, and Srivastava and of
Bencs, Berrekkal, and Regts perturb real recursion estimates to nearby
complex activities~\cite{LiuSinclairSrivastava2025,BencsBerrekkalRegts2026}.
Stronger thresholds are available with additional graph structure:
on triangle-free graphs, Liu, Sinclair, and Srivastava obtain
zero-freeness near \([0,1]\) when \(q\ge\alpha\Deg+\beta\), where
\(\alpha>\alpha_*\approx1.7632\), \(\alpha_*e^{-1/\alpha_*}=1\),
and \(\beta\) depends on \(\alpha\).  Their argument builds on the
strong spatial mixing estimates of Gamarnik, Katz, and
Misra~\cite{GamarnikKatzMisra2015}.

The reverse direction, from zero-freeness to spatial mixing, has also
seen considerable progress.  Gamarnik~\cite{Gamarnik2023} derives an
asymptotic form of spatial mixing at superlogarithmic distances.
Regts~\cite{Regts2023} and Shao and Ye~\cite{ShaoYe2024} establish
strong spatial mixing from suitable zero-free regions for graph
polynomials and spin systems, with the required closure and boundary
conditions.  For the ferromagnetic Ising model, Shao and
Shi~\cite{ShaoShi2026} obtain a Weitz-type FPTAS throughout the
Lee--Yang zero-free region without requiring strong spatial mixing on
the whole region.  Together, these results clarify when analytic
information can support local recursive approximation algorithms.

\paragraph{Couplings, spectral independence, and zero-freeness.}
Spectral independence controls the eigenvalues of conditional influence
matrices and provides a powerful route to rapid mixing, initiated by
Anari, Liu, and Oveis Gharan~\cite{AnariLiuOveisGharan2020}.
Chen, Liu, and Vigoda~\cite{ChenLiuVigoda2024} give criteria for
spectral independence based on multivariate zero-freeness.
Coupling independence supplies another such criterion: it was
formalized for Boolean distributions by Chen and
Zhang~\cite{ChenZhang2023}, and subsequently used by Chen and
Feng~\cite{ChenFeng2024} to study mixing of multi-spin systems.
The deterministic counting theorem of Chen, Feng, Guo, Zhang, and
Zou~\cite{CFFGZZ2024} gives it an additional algorithmic role.
Our transfer principle establishes a further implication, from coupling
independence to complex zero-freeness in the settings considered here.

Another recent approach obtains zero-freeness by extending Markov-chain
arguments themselves to complex parameters.  Liu, Wang, Yin, and
Yu~\cite{LiuWangYinYu2025} develop this approach for hypergraph
independent sets, and Liu and Yu~\cite{LiuYu2026} extend it to
hypergraph colourings and constraint satisfaction problems.
Our proof uses real couplings as input to a complex stability argument
across a separator, allowing the probabilistic and analytic parts of the
proof to be developed separately.

\paragraph{Independent work.}
Independently, Jingcheng Liu, Chunyang Wang, Yitong Yin, Yuichi Yoshida,
Yixiao Yu, and Zihan Zhang (personal communication) have also shown
that coupling independence implies zero-freeness.  Their results include
multivariate zero-free neighbourhoods for strictly positive spin systems,
a zero-free disc at the proper-colouring endpoint, and zero-freeness
for certain symmetric Boolean Holant problems under appropriate
coupling or adjacent-ratio assumptions.

\paragraph{Organization.}
\Cref{sec:prelim} introduces the notation and coupling independence.
\Cref{sec:transfer-principle} proves the transfer theorem, and
\Cref{sec:strict-potts} establishes the soft flip coupling and
\Cref{thm:intro-main}.  The applications to colour fields and
log-concave Holant problems appear in \Cref{sec:lee-yang,sec:holant-consequence}.
\Cref{app:additional-potts} collects further Potts applications.

\section{Preliminaries}\label{sec:prelim}
\subsection{Pinnings and partition functions}\label{sec:potts-setup}

Throughout, \(q\) and \(\Deg\) are positive integers,
\(\colours=\{1,\ldots,q\}\), \(x\in[0,1]\) is a real base point, and
\(z\in\C\) is the complex edge-activity variable.  Since \(q\) is
fixed, we suppress it from partition functions, writing, for example,
\(\nZpin{G}{\tau}(z)\) for \(\nZpin{G}{\tau}(q;z)\).  We write \(\one{A}\)
for the indicator of an event \(A\), and use the convention \(0^0=1\).
For a finite graph \(G\), write
\[
  \Delta(G):=\max_{v\in V(G)}\deg_G(v)
\]
for its maximum degree; the symbol \(\Deg\) denotes a fixed degree bound.
Let
\[
  \Gdeg:=\set{G:\text{$G$ is finite, simple, and }\Delta(G)\le\Deg},
\]
and, for \(g\ge3\), let
\[
  \Gdegg{g}:=\set{G\in\Gdeg:\operatorname{girth}(G)\ge g}.
\]
Both classes are closed under taking induced subgraphs.  We use the convention that
an acyclic graph, including the empty graph, has girth \(+\infty\).

Fix \(G=(V,E)\in\Gdeg\).  For a pinning
\(\tau\colon\Lambda\to\colours\), write
\[
  V^\tau:=V\setminus\Lambda,
  \qquad
  G^\tau:=G[V^\tau],
\]
and call \(G^\tau\) the \emph{free graph}.  For \(u\in V^\tau\) and
\(c\in\colours\), define the boundary count
\begin{equation}\label{eq:pin-counts}
  \bcount_u^\tau(c)
  :=\abs{\set{p\in\Lambda:up\in E,\ \tau_p=c}}.
\end{equation}
Boundary counts record pinned neighbours separately, even when several
have the same colour.  The normalized pinned
partition function from \eqref{eq:intro-normalized-pinned} has the factorized
form
\begin{equation}\label{eq:potts-Z}
  \nZpin{G}{\tau}(z)
  =\sum_{\varphi\colon V^\tau\to\colours}
       \prod_{\substack{uv\in E\\\{u,v\}\not\subseteq\Lambda}}
       \bigl(z\one{(\tau\cup\varphi)_u=(\tau\cup\varphi)_v}
             +\one{(\tau\cup\varphi)_u\ne(\tau\cup\varphi)_v}\bigr).
\end{equation}
Equivalently,
\[
  \nZpin{G}{\tau}(z)
  =\sum_{\varphi\colon V^\tau\to\colours}
     \prod_{u\in V^\tau}z^{\bcount_u^\tau(\varphi_u)}
     \prod_{uv\in E(G^\tau)}
       \bigl(z\one{\varphi_u=\varphi_v}+\one{\varphi_u\ne\varphi_v}\bigr).
\]

Whenever \(\nZpin{G}{\tau}(x)>0\), let \(\mupin{G}{\tau}{x}\) be the
Gibbs law on \(\colours^{V^\tau}\):
\[
  \mupin{G}{\tau}{x}(\varphi)
  :=\frac{1}{\nZpin{G}{\tau}(x)}
     \prod_{u\in V^\tau}x^{\bcount_u^\tau(\varphi_u)}
     \prod_{uv\in E(G^\tau)}
       \bigl(x\one{\varphi_u=\varphi_v}+\one{\varphi_u\ne\varphi_v}\bigr).
\]
For \(x>0\), this is the usual conditional Gibbs law: removing the
pinned-only edge contributions changes only the normalizing constant.  At \(x=0\), it is the normalized hard
pinned law, which remains meaningful even when \(\tau\) is improper.

For \(v\in V^\tau\) and \(a\in\colours\), define the \emph{child} pinning
\[
  \tau^{v=a}:=\tau\cup\{v\mapsto a\}.
\]
Then \(V^{\tau^{v=a}}=V^\tau\setminus\{v\}\), and
\begin{equation}\label{eq:potts-colour-recursion}
  \nZpin{G}{\tau}(z)
  =\sum_{a\in\colours}
     z^{\bcount_v^\tau(a)}\nZpin{G}{\tau^{v=a}}(z).
\end{equation}
The multiplier \(z^{\bcount_v^\tau(a)}\) accounts for the edges from \(v\)
to \(\Lambda\), which are pinned-only under \(\tau^{v=a}\); the normalized
function \(\nZpin{G}{\tau^{v=a}}\) omits them, so it has no forced zero at
\(z=0\) even when this multiplier vanishes there.

\subsection{Hard-colouring feasibility}
A colour \(c\) is blocked at a free vertex \(u\) when
\(\bcount_u^\tau(c)>0\).  The hard effective list of \(u\) is
\[
  \hardlist{\tau}{u}:=\set{c\in\colours:\bcount_u^\tau(c)=0}.
\]
A colouring \(\varphi\colon V^\tau\to\colours\) is \emph{hard-feasible}
if it is proper on \(G^\tau\) and
\(\varphi_u\in\hardlist{\tau}{u}\) for every \(u\in V^\tau\).

\begin{lemma}[Hard feasibility]\label{lem:hard-feas}
Assume \(q\ge\Deg+1\).  For every \(G\in\Gdeg\), every pinning \(\tau\),
and every \(u\in V^\tau\),
\[
  \abs{\hardlist{\tau}{u}}
  \ge \deg_{G^\tau}(u)+1.
\]
Consequently, \(\nZpin{G}{\tau}(0)>0\).
\end{lemma}

\begin{proof}
At most \(\abs{N_G(u)\cap\Lambda}\) colours are blocked at \(u\).  Since
\[
  \deg_{G^\tau}(u)+\abs{N_G(u)\cap\Lambda}
  =\deg_G(u)\le\Deg,
\]
we have
\[
  \abs{\hardlist{\tau}{u}}
  \ge q-\abs{N_G(u)\cap\Lambda}
  \ge \deg_{G^\tau}(u)+1.
\]
A greedy list colouring of \(G^\tau\), in any vertex ordering, leaves an
available colour at every step, so a hard-feasible colouring exists and
\(\nZpin{G}{\tau}(0)>0\).
\end{proof}

\subsection{Coupling independence and Wasserstein distance}

We fix the Wasserstein notation and define coupling independence, the
probabilistic hypothesis of \Cref{thm:potts-transfer}.

Let $(\mathcal{X},d)$ be a finite metric space. For probability measures
$\nu_0,\nu_1$ on $\mathcal{X}$,
\begin{equation}\label{eq:W1def}
  \Wd(\nu_0,\nu_1) := \inf_{\pi\in\Pi(\nu_0,\nu_1)} \E_{\pi}[d(X,Y)],
\end{equation}
where $\Pi(\nu_0,\nu_1)$ denotes the set of couplings of $\nu_0$ and $\nu_1$,
and \(X,Y\) are the two coordinate projections under \(\pi\).
On a finite space the infimum is attained.
For configurations on a common site set,
\[
  \Ham(\sigma,\sigma')
    :=\abs{\set{u:\sigma_u\ne\sigma'_u}},
  \qquad
  \WHam:=\Wd\ \text{with }d=\mathord{\Ham}.
\]

\begin{definition}[Potts coupling independence]\label{def:potts-ci}
Let \(C\in[0,\infty)\), \(\mathcal G\subseteq\Gdeg\), and
\(J\subseteq[0,1]\).  We say that
\(\mathcal G\) satisfies \(C\)-\emph{coupling independence} on \(J\) if,
for every \(G\in\mathcal G\), every pinning \(\tau\), every
\(r\in V^\tau\), all \(a,b\in\colours\), and every \(x\in J\), the two
pinned laws are well defined and
\begin{equation}\label{eq:potts-ci}
  \WHam\!\left(
    \mupin{G}{\tau^{r=a}}{x},
    \mupin{G}{\tau^{r=b}}{x}
  \right)\le C.
\end{equation}
Both laws live on the common free set \(V^\tau\setminus\{r\}\), where the
Hamming distance is taken; at \(x=0\) they are the normalized hard pinned
laws.  When
\(J=\{x\}\), we say that \(\mathcal G\) satisfies \(C\)-coupling
independence \emph{at} \(x\).
\end{definition}

For \(q\ge\Deg+1\), the normalized child law \(\mupin{G}{\tau^{r=a}}{0}\)
remains well defined even when \(a\) is blocked at \(r\), by
\Cref{lem:hard-feas}.  In that case it is defined by the normalized child
weights, rather than by conditioning the parent law \(\mupin{G}{\tau}{0}\)
on \(\varphi_r=a\); in the recursion \eqref{eq:potts-colour-recursion},
only the multiplier \(z^{\bcount_r^\tau(a)}\) of the child term vanishes
at \(z=0\).

\section{From coupling independence to zero-freeness}
\label{sec:transfer-principle}

This section proves the following precise form of
\Cref{thm:potts-transfer-informal}, with coupling independence as in
\Cref{def:potts-ci}.

\begin{theorem}[Coupling independence implies zero-freeness]
\label{thm:potts-transfer}
Fix integers \(q,\Deg\) with \(\Deg\ge2\) and \(q\ge\Deg+1\), and let
\(\mathcal G\) be a class of finite simple graphs of maximum degree at
most \(\Deg\) that is closed under taking induced subgraphs.  Assume:
\begin{enumerate}[label=\textup{(\roman*)}]
  \item \(\mathcal G\) satisfies coupling independence at \(x=0\) with
        constant \(C_0\);
  \item for every \(\delta\in(0,1]\), it satisfies coupling independence
        on \([\delta,1]\) with a finite constant \(C_\delta\).
\end{enumerate}
Then there is \(\eps>0\) such that, for every \(G\in\mathcal G\) and every
pinning \(\tau\),
\[
  \nZpin{G}{\tau}(z)\ne0
  \qquad\text{for every }z\in\mathcal U_\eps([0,1]).
\]
Consequently, \(\Zpin{G}{\tau}(z)=z^{m_G(\tau)}\nZpin{G}{\tau}(z)\) has
no zeros in the same neighbourhood except, when \(m_G(\tau)\ge1\), a zero
of multiplicity \(m_G(\tau)\) at \(z=0\).
\end{theorem}

\Cref{sec:complex-average} proves a stability estimate for complex
averages.  \Cref{sec:potts-induction} uses it in a separator induction
on each interval \([\delta,1]\): induction makes the exterior response
Lipschitz on a separator shell, and coupling independence compares the
two shell averages.  At \(0\), some interior weights vanish, so
\Cref{sec:hard-endpoint} supplies an additive-error estimate that gives a
disc around \(0\); \Cref{sec:transfer-proof-end} combines the two.

\subsection{Stability of complex averages}\label{sec:complex-average}

The separator argument writes the ratio of a normalized partition
function at \(z\) to its value at a real base point as
\(\E_\nu[e^{h+\eta}]\), where \(h\) is a common exterior response and
\(\eta\) is a small local perturbation.  We must compare two such averages
when both the law and the perturbation change.  The Lipschitz bound on
\(h\) will turn a coupling of the two laws into a bound on the change
of the exponent.  The small oscillation of \(h\), together with the
small local perturbations, will prevent complex cancellation and allow
us to compare logarithms.

\begin{lemma}[Complex averages under nearby laws]\label{lem:cavg}
Let \((\mathcal X,d)\) be a finite metric space, let \(\nu_0,\nu_1\) be
probability measures on \(\mathcal X\), and let \(A\ge0\),
\(\gamma\in[0,1/8]\).
Suppose \(h,\eta,\eta'\colon\mathcal X\to\C\) satisfy
\begin{align}
  |h(\xi)-h(\xi')|&\le A\,d(\xi,\xi'),\label{eq:cavg-lip}\\
  \max_{\xi,\xi'}|h(\xi)-h(\xi')|&\le1/8,\label{eq:cavg-range}\\
  \|\eta\|_\infty,\ \|\eta'\|_\infty&\le\gamma,\label{eq:cavg-eta}
\end{align}
where \(\|f\|_\infty=\max_\xi|f(\xi)|\).  Then the averages \(\E_{\nu_0}[e^{h+\eta}]\) and
\(\E_{\nu_1}[e^{h+\eta'}]\) are nonzero
and, with \(\Log\) of their ratio taken as the continuous branch that
vanishes at \(h=\eta=\eta'=0\) on the convex set of triples
\((h,\eta,\eta')\) satisfying
\eqref{eq:cavg-lip}--\eqref{eq:cavg-eta},
\begin{equation}\label{eq:cavg}
  \left|\Log\frac{\E_{\nu_0}[e^{h+\eta}]}{\E_{\nu_1}[e^{h+\eta'}]}\right|
  \le2A\,\Wd(\nu_0,\nu_1)+4\gamma.
\end{equation}
\end{lemma}

\begin{proof}
First subtract \(h(\xi_*)\) for a fixed \(\xi_*\in\mathcal X\).
This multiplies both averages by the same nonzero factor, so it changes
neither their ratio nor whether they vanish.  Keeping the notation \(h\)
for the centred function, \eqref{eq:cavg-range} now gives
\(\lVert h\rVert_\infty\le1/8\).

Choose an optimal coupling \((X,Y)\), with \(X\sim\nu_0\) and
\(Y\sim\nu_1\), so that
\(\E[d(X,Y)]=\Wd(\nu_0,\nu_1)\).  For \(0\le t\le1\), interpolate
the exponents on this common probability space:
\[
  u_t=t\bigl(h(X)+\eta(X)\bigr)
      +(1-t)\bigl(h(Y)+\eta'(Y)\bigr),
  \qquad M(t)=\E[e^{u_t}].
\]
The marginals identify the endpoints as
\[
  M(0)=\E_{\nu_1}[e^{h+\eta'}],
  \qquad M(1)=\E_{\nu_0}[e^{h+\eta}].
\]
Thus the required ratio is \(M(1)/M(0)\).

We first check that this path stays away from zero.  Each endpoint
exponent has absolute value at most \(1/8+\gamma\le1/4\), so the same
is true of \(u_t\).  The power-series inequality
\(|e^w-1|\le e^{|w|}-1\) yields
\[
  |M(t)-1|\le\E\bigl[|e^{u_t}-1|\bigr]
      \le e^{1/4}-1<1,
  \qquad |M(t)|\ge2-e^{1/4}.
\]
Thus the path lies in the disc \(D:=\{w\in\C:|w-1|<1\}\), and both
endpoint averages are nonzero.

Next we bound how fast the average moves.  The state space is finite,
so differentiation under the expectation gives
\[
  M'(t)=\E\!\left[e^{u_t}
    \bigl(h(X)-h(Y)+\eta(X)-\eta'(Y)\bigr)\right].
\]
The Lipschitz condition bounds the first difference by \(A\,d(X,Y)\),
and the two perturbations contribute at most \(2\gamma\).  Since
\(|e^{u_t}|\le e^{1/4}\), we obtain
\[
  |M'(t)|
  \le e^{1/4}\E\bigl[A\,d(X,Y)+2\gamma\bigr]
  =e^{1/4}\bigl(A\,\Wd(\nu_0,\nu_1)+2\gamma\bigr).
\]
This is where the coupling cost enters the estimate.

Finally, the principal logarithm \(\Log\) is analytic on \(D\), with
\(\Log1=0\).  The chain rule and the fundamental theorem of calculus give
\[
  \Log M(1)-\Log M(0)=\int_0^1\frac{M'(t)}{M(t)}\,dt.
\]
This difference is the logarithm specified in the lemma: both averages
lie in \(D\) for every admissible triple, so the difference is continuous
on the convex set of admissible triples and vanishes at
\(h=\eta=\eta'=0\).  Combining the derivative bound
with the lower bound on \(|M(t)|\), we conclude that
\[
\begin{aligned}
  \left|\Log\frac{M(1)}{M(0)}\right|
  &\le\int_0^1\frac{|M'(t)|}{|M(t)|}\,dt\\
  &\le\frac{e^{1/4}}{2-e^{1/4}}
       \bigl(A\,\Wd(\nu_0,\nu_1)+2\gamma\bigr)\\
  &\le2A\,\Wd(\nu_0,\nu_1)+4\gamma,
\end{aligned}
\]
where the last step uses \(e^{1/4}/(2-e^{1/4})<2\).
\end{proof}

In the diagram below, each arrow label bounds the modulus of the
logarithm of the ratio of its endpoint averages.  The horizontal solid arrow uses
\(\eta=\eta'=0\).  Along a vertical arrow the law stays fixed, so we
take \(X=Y\), and one endpoint has zero perturbation.  The perturbation
difference in the derivative bound is then at most \(\gamma\), giving
\(2\gamma\).
The dashed arrow is \eqref{eq:cavg}.  Thus coupling independence enters
through the horizontal Wasserstein term, while the local perturbations
account for the two vertical terms.
\[
  \begin{tikzcd}[column sep=8em,row sep=3.2em,labels={font=\footnotesize}]
    \E_{\nu_0}[e^h]
      \arrow[d, "{2\gamma}"]
      &
    \E_{\nu_1}[e^h]
      \arrow[l, "{2A\,\Wd(\nu_0,\nu_1)}"']
      \\
    \E_{\nu_0}[e^{h+\eta}]
      &
    \E_{\nu_1}[e^{h+\eta'}]
      \arrow[u, "{2\gamma}"']
      \arrow[l, dashed, "{2A\,\Wd(\nu_0,\nu_1)+4\gamma}"']
  \end{tikzcd}
\]

\subsection{Positive-activity estimates}
\label{sec:potts-induction}

The estimates in this subsection are uniform on every compact interval
\([\delta,1]\subset(0,1]\).
Fix \(G\in\Gdeg\), a pinning \(\tau\), and a free vertex \(v\).  All
graph distances in this subsection are measured in the free graph \(G^\tau\).
For \(r\ge0\), write
\[
  B_r(v):=\{u\in V^\tau:\dist_{G^\tau}(u,v)\le r\},
  \qquad
  S_r(v):=\{u\in V^\tau:\dist_{G^\tau}(u,v)=r\}.
\]
For colours \(a,b\) and a real base point \(x\), let \(\nu_a^{(r)}\) and
\(\nu_b^{(r)}\) be the \(S_r(v)\)-marginals of
\(\mupin{G}{\tau^{v=a}}{x}\) and \(\mupin{G}{\tau^{v=b}}{x}\).  If the
two pinned laws admit a coupling of expected Hamming cost at most \(C\),
then, for every integer \(L\ge1\),
\begin{equation}\label{eq:potts-shell-sum}
  \sum_{r=2}^{L+1}
    \WHam\bigl(\nu_a^{(r)},\nu_b^{(r)}\bigr)
  \le C.
\end{equation}
Indeed, fix a single coupling \((X,Y)\) of the two pinned laws with
expected Hamming cost at most \(C\).  Projecting this coupling onto each
shell gives
\[
  \WHam\bigl(\nu_a^{(r)},\nu_b^{(r)}\bigr)
  \le\E\!\left[\Ham\bigl(X_{S_r(v)},Y_{S_r(v)}\bigr)\right].
\]
Since the shells are disjoint, the right-hand sides sum to at most
\(\E[\Ham(X,Y)]\le C\), which proves \eqref{eq:potts-shell-sum}.
Hence one of these \(L\) shells has cost at most \(C/L\).  If a
shell is empty, the component of \(v\) has size bounded in terms of
\((\Deg,L)\) and is treated as a finite local system.

Choose \(r\ge2\) and put
\[
  S:=S_r(v),\qquad
  U:=B_{r-1}(v)\setminus\{v\},\qquad
  G':=G-B_{r-1}(v);
\]
for \(\xi\colon S\to\colours\), write \(\tau_\xi:=\tau\cup\xi\).
For \(c\in\colours\), split a colouring at \(S\) as follows.  Among the
edge factors present in the normalized child \(\nZpin{G}{\tau^{v=c}}\),
assign to the exterior precisely those with an endpoint in
\(V^\tau\setminus B_r(v)\).  Let
\(D_{c,\xi}(z)\) be the sum, over colourings of \(U\), of the product of
the remaining factors, evaluated with \(v=c\) and \(S=\xi\).  In particular, factors from edges within \(S\)
and from edges joining \(S\) to \(\Lambda\) belong to \(D_{c,\xi}\).
Edges joining \(v\) to \(\Lambda\) are pinned-only in the child, so they
contribute no factor here.  Since no edge crosses directly from
\(U\) to the exterior, we obtain the \emph{separator identity}
\begin{equation}\label{eq:potts-separator-defined}
  \nZpin{G}{\tau^{v=c}}(z)
  =\sum_{\xi\colon S\to\colours}
      D_{c,\xi}(z)\nZpin{G'}{\tau_\xi}(z).
\end{equation}
At a real base point \(x\), the shell marginal \(\nu_c=\nu_c^{(r)}\) is
\begin{equation}\label{eq:potts-shell-marginal-defined}
  \nu_c(\xi)
  :=\frac{D_{c,\xi}(x)\nZpin{G'}{\tau_\xi}(x)}
          {\nZpin{G}{\tau^{v=c}}(x)}.
\end{equation}
The exterior instance \((G',\tau_\xi)\) does not depend on the root
colour \(c\), and changing one coordinate of \(\xi\) changes only one
pinned colour in it; this is where the inductive one-vertex response bound
enters.
\begin{lemma}[Positive-activity Potts induction]
\label{lem:potts-positive-response}
Fix integers \(q,\Deg\) with \(\Deg\ge2\), let \(\delta\in(0,1]\) and
\(C\in[0,\infty)\), and let
\(\mathcal G\subseteq\Gdeg\) be closed under taking induced subgraphs and satisfy
\(C\)-coupling independence on \([\delta,1]\).  There are
\(\alpha,\eps>0\), depending only on \((q,\Deg,\delta,C)\), such that, for
every \(G\in\mathcal G\), every pinning \(\tau\), every
\(x\in[\delta,1]\), and every \(z\in\C\) with \(|z-x|<\eps\),
\[
  \nZpin{G}{\tau}(z)\ne0,
\]
and, for every \(v\in V^\tau\) and \(a,b\in\colours\),
\begin{equation}\label{eq:potts-response-invariant}
  \left|\Log
  \frac{\nZpin{G}{\tau^{v=a}}(z)/\nZpin{G}{\tau^{v=a}}(x)}
       {\nZpin{G}{\tau^{v=b}}(z)/\nZpin{G}{\tau^{v=b}}(x)}\right|
  \le\alpha.
\end{equation}
The logarithm is continued from \(z=x\).
\end{lemma}

\begin{proof}
\textbf{Local bounds.}\hspace{1em}Enlarging \(C\) if necessary, assume \(C\ge1\), and take
\(L=\lceil64C\rceil\).  Set
\[
  s:=\Deg(\Deg-1)^L,
  \qquad
  N:=1+\Deg\sum_{k=0}^{L-1}(\Deg-1)^k.
\]
For \(2\le r\le L+1\), breadth-first bounds give
\( |S_r(v)|\le\Deg(\Deg-1)^{r-1}\le s\) and
\( |B_{r-1}(v)|\le N\).  Choose
\[
  0<\alpha\le\min\left\{\frac\pi{16},\frac1{16s}\right\}.
\]
Here \(L\) provides enough disjoint shells to reduce the coupling cost,
\(s\) bounds a shell, \(N\) separates the bounded-component case, and
\(\alpha\) is the logarithmic and angular budget.
Up to relabelling, only finitely many interior polynomials \(D_{c,\xi}\)
and polynomials of free components with at most \(N\) vertices occur: each is
determined by a graph of maximum degree at most \(\Deg\) on at most
\(N+s\) vertices, the root, the shell, the colours \(c,\xi\), and
boundary counts in \(\{0,\ldots,\Deg\}\).  They are positive on the compact
interval \([\delta,1]\).  We may therefore choose \(\eps>0\), uniformly in
\(x\in[\delta,1]\), so that they remain nonzero for \(|z-x|<\eps\),
the normalized child quotients of components of size at most \(N\)
satisfy \eqref{eq:potts-response-invariant}, and
\begin{equation}\label{eq:potts-local-budget}
  \left|\Log\frac{D_{c,\xi}(z)}{D_{c,\xi}(x)}\right|
  \le\frac\alpha8,
  \qquad
  \left|\Im\Log\frac{z^m}{x^m}\right|+\alpha<\frac\pi2
  \quad(0\le m\le\Deg).
\end{equation}
All logarithms continue from \(z=x\).

\paragraph{Exterior response.}
Induct on the number of free vertices, proving zero-freeness and
\eqref{eq:potts-response-invariant} simultaneously.  Fix \(G,\tau,v,a,b\).
If the free component of \(v\) has size at most \(N\), cancel the other
components; the uniform local choice of \(\eps\) gives the response bound
directly.  If it is larger, \eqref{eq:potts-shell-sum} supplies a shell
\(S\) with
\[
  \WHam(\nu_a,\nu_b)\le C/L.
\]
Induction on the common exterior defines
\[
  h_\xi:=\Log
  \frac{\nZpin{G'}{\tau_\xi}(z)}
       {\nZpin{G'}{\tau_\xi}(x)}.
\]
We write \(h(\xi):=h_\xi\).  When \(\xi\) and \(\xi'\) differ at one shell vertex \(w\), put
\(\tau':=\tau\cup\xi|_{S\setminus\{w\}}\).  Then \(G'\) is an induced
subgraph of \(G\), so \(G'\in\mathcal G\); the instance \((G',\tau')\) has
\(|V^\tau|-|B_r(v)|+1<|V^\tau|\) free vertices, and its two one-vertex
pinnings at \(w\) are exactly \((G',\tau_\xi)\) and \((G',\tau_{\xi'})\);
pinned-only factors are omitted by \eqref{eq:potts-Z}.  Thus the
inductive response bound applies to \(w\) in \((G',\tau')\).  Both
\(h_\xi-h_{\xi'}\) and that response logarithm are analytic and vanish at
\(z=x\), so they are the same continued branch.  Summing along a Hamming
path therefore gives
\[
  |h_\xi-h_{\xi'}|
  \le\alpha\Ham(\xi,\xi')
  \le s\alpha
  \le\frac1{16}.
\]
\paragraph{Comparison of shell averages.}
The bounded interior has normalized change
\[
  \eta_c(\xi):=\Log\frac{D_{c,\xi}(z)}{D_{c,\xi}(x)},
  \qquad
  \|\eta_c\|_\infty\le\frac\alpha8.
\]
The separator identity is now exactly
\begin{equation}\label{eq:potts-separator-expectation}
  \frac{\nZpin{G}{\tau^{v=c}}(z)}
       {\nZpin{G}{\tau^{v=c}}(x)}
  =\E_{\nu_c}[e^{h+\eta_c}].
\end{equation}
The hypotheses of \Cref{lem:cavg}, with \(A=\alpha\) and
\(\gamma=\alpha/8\), hold at every \(z'\) with \(|z'-x|<\eps\), and
\(h,\eta_a,\eta_b\) vanish at \(z'=x\).  Composing the continuous branch
of \Cref{lem:cavg} with \(z'\mapsto(h,\eta_a,\eta_b)\) therefore gives a
continuous logarithm of the ratio in \eqref{eq:potts-response-invariant}
that vanishes at \(z'=x\), namely the logarithm continued from \(z=x\).
Thus \Cref{lem:cavg} gives
\[
  \left|\Log
  \frac{\nZpin{G}{\tau^{v=a}}(z)/\nZpin{G}{\tau^{v=a}}(x)}
       {\nZpin{G}{\tau^{v=b}}(z)/\nZpin{G}{\tau^{v=b}}(x)}\right|
  \le2\alpha C/L+\frac\alpha2<\alpha.
\]

\paragraph{Parent nonvanishing.}
To deduce nonvanishing of the parent function \(\nZpin{G}{\tau}\), put
\[
  R_c:=\frac{\nZpin{G}{\tau^{v=c}}(z)}
             {\nZpin{G}{\tau^{v=c}}(x)}
\]
and choose any \(a_0\in\colours\).  The colour recursion
\eqref{eq:potts-colour-recursion} gives
\[
  \nZpin{G}{\tau}(z)
  =R_{a_0}
   \sum_a x^{\bcount_v^\tau(a)}\nZpin{G}{\tau^{v=a}}(x)
   \exp\!\left(
      \Log\frac{z^{\bcount_v^\tau(a)}}{x^{\bcount_v^\tau(a)}}
      +\Log\frac{R_a}{R_{a_0}}\right).
\]
The coefficients are positive.  By \eqref{eq:potts-local-budget} and the
response bound, every exponential has positive real part,
so the sum lies in an open half-plane and is nonzero.  This closes the
induction.
\end{proof}

\subsection{Hard-endpoint estimates}
\label{sec:hard-endpoint}

At \(z=0\), a boundary weight may vanish.  Hard-feasible colourings contribute
positive weight at \(z=0\), while the term of every colouring that is not
hard-feasible is divisible by \(z\).

\begin{lemma}[One-vertex comparison at \texorpdfstring{\(z=0\)}{z=0}]\label{lem:potts-hard-comparison}
Assume \(q\ge\Deg+1\) and put \(H_*:=q^\Deg\).  For every
\(G\in\Gdeg\), every pinning \(\tau\), every \(y\in V^\tau\), and all
colours \(a,b\),
\begin{equation}\label{eq:potts-hard-comparison}
  \nZpin{G}{\tau^{y=a}}(0)
  \le H_*\nZpin{G}{\tau^{y=b}}(0).
\end{equation}
\end{lemma}

\begin{proof}
Measure distance in \(G^\tau\), condition on \(S_2(y)\), and cancel free
components outside that of \(y\).  For each shell colouring compatible
with the exterior hard constraints, the exterior factor is the same for
the two colours at \(y\).  Only the at most \(\Deg\) vertices in
\(S_1(y)\) remain to be coloured.  Fix such a shell colouring \(\xi\)
and a root colour \(c\in\{a,b\}\).  For every \(u\in S_1(y)\), after
forbidding \(c\) and the colours used by its neighbours in \(S_2(y)\),
the remaining list has size at least
\[
\begin{aligned}
 |\hardlist{\tau}{u}|-1-|N_{G^\tau}(u)\cap S_2(y)|
 &\ge \deg_{G^\tau}(u)+1-1-|N_{G^\tau}(u)\cap S_2(y)|\\
 &=\deg_{G^\tau[S_1(y)]}(u)+1.
\end{aligned}
\]
Here the equality uses that every \(u\in S_1(y)\) is adjacent to \(y\),
and its other neighbours lie in \(S_1(y)\cup S_2(y)\).  Greedy
list-colouring of \(G^\tau[S_1(y)]\) therefore gives at least one extension
for each of the two root colours.  There are at most
\(q^{|S_1(y)|}\le q^\Deg\) extensions.  Comparing term by term over the
shell proves the claim.
\end{proof}

\begin{proposition}[Potts disc around \texorpdfstring{\(z=0\)}{z=0}]\label{prop:ep-disc}
Fix integers \(q,\Deg\) with \(\Deg\ge2\) and \(q\ge\Deg+1\), let
\(C_0\in[0,\infty)\), and let
\(\mathcal G\subseteq\Gdeg\) be closed under taking induced subgraphs and satisfy
\(C_0\)-coupling independence at \(x=0\).  There is
\(\rho_0=\rho_0(q,\Deg,C_0)>0\) such that
\[
  \nZpin{G}{\tau}(z)\ne0
  \qquad(|z|<\rho_0)
\]
for every \(G\in\mathcal G\) and every pinning \(\tau\).
\end{proposition}

\begin{proof}
\textbf{Setup and shell choice.}\hspace{1em}Enlarging \(C_0\) if necessary, assume \(C_0\ge1\), put
\(L=\lceil64C_0\rceil\), and set
\[
  s:=\Deg(\Deg-1)^L,
  \qquad
  N:=1+\Deg\sum_{k=0}^{L-1}(\Deg-1)^k.
\]
For every vertex \(v\) and \(2\le r\le L+1\), these quantities bound
\(|S_r(v)|\) and \(|B_{r-1}(v)|\), respectively.  Choose
\[
  0<\alpha\le\min\left\{\frac\pi{16},\frac1{16s}\right\}.
\]
These constants have the same roles as in the proof of
\Cref{lem:potts-positive-response}: \(L\) yields a low-cost
shell, \(s\) bounds its Hamming diameter, \(N\) isolates bounded
components, and \(\alpha\) is the logarithmic budget.  The only new step is
to separate the hard-feasible contribution from an additive error divisible
by \(z\).
Below \(K\) denotes a positive constant depending only on
\((q,\Deg,C_0)\), and may increase from line to line.  We prove, for a
sufficiently small uniform \(\rho_0\) and all \(|z|<\rho_0\),
zero-freeness and
\begin{equation}\label{eq:potts-hard-response}
  \left|\Log
  \frac{\nZpin{G}{\tau^{v=a}}(z)/\nZpin{G}{\tau^{v=a}}(0)}
       {\nZpin{G}{\tau^{v=b}}(z)/\nZpin{G}{\tau^{v=b}}(0)}\right|
  \le\alpha
\end{equation}
for every \(G\in\mathcal G\), every pinning \(\tau\), every
\(v\in V^\tau\), and all \(a,b\).  The denominators are positive by
\Cref{lem:hard-feas}.  Throughout this proof, endpoint logarithms are the
analytic branches continued from \(z=0\) and normalized there by
\(\Log 1=0\).  Induct on the number of free vertices.  In the current
inductive instance, fix \(G,\tau,v\), measure distance in \(G^\tau\), and
use the sets \(B_r(v),S_r(v)\) and the separator data
\(S,U,G',\tau_\xi,D_{c,\xi}\) from
\eqref{eq:potts-separator-defined}--\eqref{eq:potts-shell-marginal-defined},
now at the real base point \(x=0\).

If the free component containing \(v\) has at most \(N\) vertices, the
other components cancel in \eqref{eq:potts-hard-response}.  Each
remaining component factor \(F\) satisfies \(F(0)\ge1\) by
\Cref{lem:hard-feas} and \(|F(z)-F(0)|\le q^N|z|\) for \(|z|\le1\), so
\eqref{eq:potts-hard-response} holds, with these factors nonzero,
whenever \(|z|<\alpha/(4q^N)\); we take \(\rho_0\le\alpha/(4q^N)\).
For a larger component, choose a shell with
\(\WHam(\nu_a,\nu_b)\le C_0/L\).

\paragraph{Exterior response.}
Induction on the common exterior defines
\[
  h_\xi:=\Log
  \frac{\nZpin{G'}{\tau_\xi}(z)}
       {\nZpin{G'}{\tau_\xi}(0)}.
\]
For shell colourings \(\xi,\xi'\) that differ only at a shell vertex
\(w\), regard all other shell vertices as part of the base pinning and
temporarily unpin \(w\).
This realizes the two normalized exterior partition functions as the
two one-vertex pinnings
of the same induced instance; pinned-only factors are omitted by
\eqref{eq:potts-Z}.  Hence the inductive response bound
\eqref{eq:potts-hard-response} applies to this smaller instance.  Branch
uniqueness at \(z=0\) then identifies \(h_\xi-h_{\xi'}\) with the inductive
one-vertex response logarithm.  Summing along a Hamming path gives
\[
  |h_\xi-h_{\xi'}|
  \le\alpha\Ham(\xi,\xi')
  \le s\alpha
  \le\frac1{16}.
\]
\paragraph{Hard-feasible contribution and error control.}
Here \(\nu_c\) is the shell marginal at \(x=0\).  For each \(c\), choose a
hard-feasible colouring of the free vertices under \(\tau^{v=c}\), and let
\(\xi_c\) be its shell restriction.
This witness supplies at least one hard interior extension, so
\[
  D_{c,\xi_c}(0)\ge1,
  \qquad
  \nZpin{G}{\tau^{v=c}}(0)
  \ge \nZpin{G'}{\tau_{\xi_c}}(0).
\]
For a fixed interior colouring \(\varphi\), its contribution is
\(z^{m(\varphi)}\) for some \(m(\varphi)\ge0\).  After subtracting the value
at \(0\), this contribution is zero if \(m(\varphi)=0\) and is
\(z^{m(\varphi)}\) otherwise.  Since there are at most \(q^N\) interior
colourings,
\[
  |D_{c,\xi}(z)-D_{c,\xi}(0)|\le q^N|z|
  \qquad(|z|\le1/2).
\]

Splitting \(D_{c,\xi}(z)\) into its value at \(0\) and its error in the
separator identity gives
\[
  R_c:=\frac{\nZpin{G}{\tau^{v=c}}(z)}
             {\nZpin{G}{\tau^{v=c}}(0)}
      =M_c+E_c,
  \qquad
  M_c:=\E_{\nu_c}[e^h],
\]
where, explicitly,
\[
  E_c=
  \frac{
    \sum_\xi
      \bigl(D_{c,\xi}(z)-D_{c,\xi}(0)\bigr)
      \nZpin{G'}{\tau_\xi}(0)e^{h_\xi}
  }{
      \nZpin{G}{\tau^{v=c}}(0)
  }.
\]
Changing the shell colours one at a time and applying
\Cref{lem:potts-hard-comparison} at most \(s\) times gives
\[
  \nZpin{G'}{\tau_\xi}(0)
  \le H_*^s\nZpin{G'}{\tau_{\xi_c}}(0).
\]
There are at most \(q^s\) shell colourings, and
\(|h_\xi-h_{\xi_c}|\le1/16\).  The preceding bounds therefore give
\[
  |E_c|
  \le q^{N+s}H_*^s e^{1/16}|z|\,|e^{h_{\xi_c}}|.
\]
For the main term,
\[
\begin{aligned}
  \left|\frac{M_c}{e^{h_{\xi_c}}}-1\right|
  &\le \sum_\xi\nu_c(\xi)
      \left|e^{h_\xi-h_{\xi_c}}-1\right|\\
  &\le e^{1/16}-1<\frac14.
\end{aligned}
\]
Hence
\[
  |M_c|\ge(2-e^{1/16})|e^{h_{\xi_c}}|
  >\frac34|e^{h_{\xi_c}}|.
\]
These two estimates give
\(|E_c/M_c|\le K|z|\).  Shrink \(\rho_0\), if necessary, so that
\(\rho_0\le1/2\), \(K\rho_0\le1/2\), and
\(4K\rho_0<\alpha/2\).  Then
\[
  \left|\Log(1+E_c/M_c)\right|
  \le2K|z|,
\]
with the logarithm continued from \(z=0\).

Apply \Cref{lem:cavg} with \(A=\alpha\) and \(\eta=\eta'=0\).  Its
hypotheses hold at every \(z'\) with \(|z'|<\rho_0\), and \(h\) vanishes
at \(z'=0\); as in the proof of \Cref{lem:potts-positive-response}, the
branch in \Cref{lem:cavg} is therefore the logarithm of \(M_a/M_b\)
continued from \(z=0\).  The lemma gives
\[
  \left|\Log\frac{M_a}{M_b}\right|\le2\alpha C_0/L.
\]
We have
\[
  \Log\frac{R_a}{R_b}
  =\Log\frac{M_a}{M_b}
   +\Log(1+E_a/M_a)-\Log(1+E_b/M_b).
\]
Combining this identity with \(L\ge64C_0\) and
\(4K|z|<\alpha/2\), we find that the absolute value of its right-hand
side is at most \(\alpha/32+\alpha/2<\alpha\).  This proves
\eqref{eq:potts-hard-response}.

\paragraph{Parent nonvanishing.}
To close the induction, we use the parent recursion
\[
  \nZpin{G}{\tau}(z)
  =\sum_c z^{\bcount_v^\tau(c)}\nZpin{G}{\tau^{v=c}}(z).
\]
Choose a hard-feasible colouring \(\sigma\) and put \(a_0=\sigma_v\), so
\(\bcount_v^\tau(a_0)=0\).  After division by \(R_{a_0}\), the terms with
\(\bcount_v^\tau(c)=0\) have positive real part.  The \(a_0\)-term is
exactly \(\nZpin{G}{\tau^{v=a_0}}(0)\), so their sum has real part at
least this value.  The remaining terms are
divisible by \(z\).  Here
\(\abs{R_c/R_{a_0}}\le e^\alpha\), and \Cref{lem:potts-hard-comparison}
gives \(\nZpin{G}{\tau^{v=c}}(0)\le H_*\nZpin{G}{\tau^{v=a_0}}(0)\);
hence the at most \(q\) remaining terms have total modulus at most
\(K|z|\nZpin{G}{\tau^{v=a_0}}(0)\).  A final shrink of
\(\rho_0\) makes this error smaller than the hard-feasible contribution
and closes the induction.
\end{proof}

\subsection{Proof of the transfer theorem}
\label{sec:transfer-proof-end}

\begin{proof}[Proof of \Cref{thm:potts-transfer}]
Let \(\rho_0\) come from \Cref{prop:ep-disc}, set
\[
  \delta_0:=\min\{1/2,\rho_0/4\},
\]
and apply \Cref{lem:potts-positive-response} with
\(\delta=\delta_0\) and \(C=C_{\delta_0}\).  Let \(\eps_+\) be the
resulting uniform radius and put \(\eps:=\min\{\eps_+,\delta_0\}\).
Let \(z\in\mathcal U_\eps([0,1])\), and let \(x\in[0,1]\) be a nearest
point to \(z\), so that \(|z-x|<\eps\).  If \(x\le\delta_0\), then
\(|z|<2\delta_0<\rho_0\); otherwise \(x\in[\delta_0,1]\) and the
positive-activity induction applies.  The statement for \(\Zpin{G}{\tau}\)
then follows from \eqref{eq:intro-normalization}.
\end{proof}

\section{\texorpdfstring{Coupling independence for \(q\ge11\Deg/6\)}
{Coupling independence for q >= 11 Delta / 6}}
\label{sec:strict-potts}

This section supplies the coupling input to \Cref{thm:potts-transfer} and
proves \Cref{thm:intro-main}.  Above the threshold, it establishes the
first step of the implication
\[
  q>\frac{11}{6}\Deg
  \quad\Longrightarrow\quad
  \text{coupling independence on \([0,1]\)}
  \quad\Longrightarrow\quad
  \text{zero-freeness near \([0,1]\)},
\]
whose second step is \Cref{thm:potts-transfer}.  We obtain the first step
by extending Vigoda's flip coupling~\cite{Vigoda2000} to positive
activities, conditioning on random hard constraints.  At equality, we
combine the positive-activity coupling estimates with the hard-colouring
coupling-independence bound of~\cite[Theorem~20]{CFFGZZ2024}; the
normalization and pinning conventions are checked in
\Cref{cor:critical-line-input}.  Only the equality case \(q=11\Deg/6\)
of \Cref{thm:intro-main} uses this external input.

\begin{theorem}[Hamming coupling independence on \texorpdfstring{$[0,1]$}{[0,1]}]
\label{thm:strict-ci}
Let \(q,\Deg\) be positive integers with \(\Deg\ge2\) and
\(q>11\Deg/6\).  Put
\[
  \gapq:=\frac q\Deg-\frac{11}{6}>0.
\]
Then, for every \(G\in\Gdeg\), every pinning \(\tau\), every
\(x\in[0,1]\), every \(r\in V^\tau\), and all
\(a,b\in\colours\),
\[
  \WHam\!\left(
    \mupin{G}{\tau^{r=a}}{x},
    \mupin{G}{\tau^{r=b}}{x}
  \right)
  \le
  \frac{2(1-x)\Deg}{q-(11/6)(1-x)\Deg}
  \le\frac2{\gapq}.
\]
At \(x=0\) both laws are well defined by \Cref{lem:hard-feas}, since
the hypotheses imply \(q\ge\Deg+2\).
\end{theorem}

To compare the laws with two different root colours, we compare each to
the common system obtained by deleting the root.  Pinning the root adds at
most \(\Deg\) boundary constraints to this common system.  We need two estimates: how much these
constraints change one update, and how strongly the kernel contracts.
Stationary comparison converts them into a distance between stationary
laws; the triangle inequality then gives coupling independence.
\Cref{sec:stationary-comparison} records stationary comparison and the
endpoint continuity that covers \(x=0\);
\Cref{sec:soft-flip-interface,sec:boundary-sensitivity} construct the
kernels and bound the boundary effect; \Cref{sec:hard-one-step} proves
the conditional one-step estimate; \Cref{sec:activation-average} averages
it and applies stationary comparison, proving
\Cref{thm:strict-ci,cor:critical-line-input}; and \Cref{sec:main-proof}
proves \Cref{thm:intro-main}.

\subsection{Stationary comparison and endpoint continuity}
\label{sec:stationary-comparison}

For a Markov kernel \(P\), write \(\nu P\) for the law after one update
from \(\nu\).  We use the following variant of
\cite[Lemma~19]{CFFGZZ2024}, in which the contracting chain is called
\(P\), contraction is assumed only for adjacent states of a path metric,
and the disagreement is averaged under \(\nu_Q\) rather than bounded
pointwise.

\begin{lemma}[Stationary comparison]\label{lem:stationary-comparison}
Let \(d\) be the shortest-path metric of a finite connected state graph
with positive edge lengths.  Let \(P,Q\) be Markov kernels on the vertex
set of this graph, with stationary laws
\(\nu_P,\nu_Q\).  Suppose \(0<\kappa\le1\), \(C\ge0\), and
\[
  \Wd(P(\xi,\cdot),P(\xi',\cdot))\le(1-\kappa)d(\xi,\xi')
  \quad\text{for adjacent }\xi,\xi',
  \qquad
  \E_{\xi\sim\nu_Q}\Wd(P(\xi,\cdot),Q(\xi,\cdot))\le C.
\]
Then \(\Wd(\nu_P,\nu_Q)\le C/\kappa\).
\end{lemma}

\begin{proof}
Summing along shortest paths extends the first bound to all states;
coupling the initial states then gives
\(\Wd(\nu P,\nu' P)\le(1-\kappa)\Wd(\nu,\nu')\) for all laws \(\nu,\nu'\).
Stationarity and the second hypothesis yield
\[
  \Wd(\nu_P,\nu_Q)
  \le\Wd(\nu_PP,\nu_QP)+\Wd(\nu_QP,\nu_QQ)
  \le(1-\kappa)\Wd(\nu_P,\nu_Q)+C.
\]
Rearranging gives the claim.
\end{proof}

\begin{lemma}[Finite-state continuity of normalized laws]
\label{lem:finite-endpoint-closure}
On a finite space \(\Omega\), let two families of nonnegative weights be
continuous at \(x=0\), with positive total weights there.  Their normalized
laws \(\mu^{(0)}_x,\mu^{(1)}_x\) converge in total variation to
\(\mu^{(0)}_0,\mu^{(1)}_0\) as
\(x\downarrow0\), and for every metric \(d\) on \(\Omega\),
\[
  \Wd(\mu^{(0)}_x,\mu^{(1)}_x)\longrightarrow
  \Wd(\mu^{(0)}_0,\mu^{(1)}_0).
\]
Consequently, if \(\Wd(\mu^{(0)}_x,\mu^{(1)}_x)\le B(x)\) for all small
\(x>0\), then \(\Wd(\mu^{(0)}_0,\mu^{(1)}_0)\le\liminf_{x\downarrow0}B(x)\).
\end{lemma}

\begin{proof}
The normalized weights converge pointwise.  Use
\(\Wd(\nu,\nu')\le\operatorname{diam}_d(\Omega)
\|\nu-\nu'\|_{\rm TV}\) and the triangle inequality.
\end{proof}

\subsection{Random hard constraints and the soft flip kernel}
\label{sec:soft-flip-interface}

Fix \(G\in\Gdeg\), a pinning \(\tau\), and \(x\in(0,1]\).
Let \(\mathcal E_\tau\) consist of the free edges and the individually
labelled free--pinned edges; pinned-only edges have already been removed
by normalization.  Each member imposes an inequality between its endpoint
colours.  The identity
\begin{equation}\label{eq:active-identity}
  x\one{c=c'}+\one{c\ne c'}=x+(1-x)\one{c\ne c'}
\end{equation}
gives the following representation.  Given \(\sigma\), activate each
satisfied inequality independently with probability \(1-x\).
For the active set \(A\), let \(E_A\) consist of its free edges and put
\[
  L_{A,u}=\colours\setminus
    \{\tau_p:up\in A\text{ is a free--pinned edge}\},
  \qquad F_A=(V^\tau,E_A,L_A).
\]
For any hard list instance \(F=(W,E,L)\), write \(\Omega_F\) for its
proper colourings respecting every list.

\begin{lemma}[Soft Potts as a random hard-colouring model]
\label{lem:active-rep}
If \(\sigma\sim\mupin{G}{\tau}{x}\) and \(A\) is sampled as above,
then \(\sigma\), conditional on \(A\), is uniform on \(\Omega_{F_A}\).
\end{lemma}

\begin{proof}
For \(M=|\mathcal E_\tau|\), the joint unnormalized weight is
\begin{equation}\label{eq:joint-active-law}
  x^{M-|A|}(1-x)^{|A|}\one{\sigma\in\Omega_{F_A}}.
\end{equation}
Summing over \(A\) gives a factor \(1\) for each satisfied constraint
and \(x\) for each unsatisfied one, hence the Potts weight.  For fixed
\(A\), the displayed weight is constant on \(\Omega_{F_A}\).
\end{proof}

A \emph{flip profile} is a sequence \(\rho=(\rho_s)_{s\ge1}\) with
\(0\le\rho_s\le1\) for all \(s\).

\begin{definition}[Hard list-flip transition]\label{def:flip}
From \(X\in\Omega_F\), choose \((u,c)\) uniformly in \(W\times\colours\).
If \(c=X_u\), hold.  Otherwise, let \(S_X(u,c)\) be the
\(\{X_u,c\}\)-component of \(u\), that is, the component containing \(u\)
of the subgraph of \((W,E)\) induced by the vertices coloured \(X_u\) or
\(c\).  Swap its two colours with probability
\(\rho_{|S_X(u,c)|}/|S_X(u,c)|\), provided the swap respects every list;
otherwise hold.  Denote this kernel by \(\Phi_F^\rho\), taking the
identity kernel when \(W=\varnothing\).
\end{definition}

A feasible component swap of size \(s\) has \(s\) proposals, hence total
probability \(\rho_s/(|W|q)\).  The swap is an involution with the same
reverse rate, so \(\Phi_F^\rho\) is reversible for the uniform law.

\begin{definition}[Soft flip kernel]\label{def:soft-kernel}
The kernel \(K_{x,\rho}^{G,\tau}\) samples \(A\) conditional on the
current colouring, applies \(\Phi_{F_A}^\rho\) (\Cref{def:flip}), and
returns the updated
colouring to \(G\), keeping \(\tau\) fixed (\Cref{fig:soft-kernel-step}).
\end{definition}

\begin{figure}[t]
\centering
\begin{tikzpicture}[
  x=1cm,y=1cm,
  vertex/.style={
    circle,draw=black!55,line width=.65pt,minimum size=6.4mm,
    inner sep=0pt,font=\small\bfseries,text=black!85
  },
  one/.style={vertex,fill=blue!13},
  two/.style={vertex,fill=orange!24},
  three/.style={vertex,fill=teal!17},
  four/.style={vertex,fill=violet!17},
  full/.style={draw=black!48,line width=.7pt},
  active/.style={draw=teal!75!black,line width=1.1pt},
  activevertex/.style={draw=teal!75!black,line width=1.05pt},
  inactive/.style={
    draw=black!35,dash pattern=on 2pt off 2pt,line width=.55pt
  },
  component/.style={draw=teal!9,line width=9mm,line cap=round},
  flowarrow/.style={
    -{Latex[length=1.7mm,width=1.2mm]},draw=black!65,line width=.7pt
  },
  paneltitle/.style={font=\normalsize,anchor=base},
  flowlabel/.style={font=\small,text=black!75,inner sep=2pt},
  legend/.style={font=\footnotesize,text=black!70,anchor=west,inner sep=2pt},
  pics/softstate/.style n args={2}{code={
    \coordinate (a) at (-.82,.82);
    \coordinate (b) at ( .82,.82);
    \coordinate (c) at (0,0);
    \coordinate (d) at (-.82,-.82);
    \coordinate (e) at ( .82,-.82);
    \ifnum#1=1
      \draw[component] (a)--(c)--(e);
      \foreach \u/\v in {a/b,a/d,d/e,b/e,b/c,c/d}
        {\draw[inactive] (\u)--(\v);}
      \draw[active] (a)--(c)--(e);
    \else
      \foreach \u/\v in {a/b,a/d,d/e,b/e,a/c,b/c,c/d,c/e}
        {\draw[full] (\u)--(\v);}
    \fi
    \node[three] at (b) {3};
    \node[four] at (d) {4};
    \begin{scope}
      \ifnum#1=1
        \tikzset{vertex/.append style={activevertex}}
      \fi
      \ifnum#2=1
        \node[two] at (a) {2};
        \node[one] at (c) {1};
        \node[two] at (e) {2};
      \else
        \node[one] at (a) {1};
        \node[two] at (c) {2};
        \node[one] at (e) {1};
      \fi
    \end{scope}
  }}
]
  \pic at (0,0) {softstate={0}{0}};
  \pic at (4.1,0) {softstate={1}{0}};
  \pic at (8.2,0) {softstate={1}{1}};
  \pic at (12.3,0) {softstate={0}{1}};

  \node[paneltitle] at (0,1.43) {\((G,\sigma)\)};
  \node[paneltitle] at (4.1,1.43) {\((F_A,\sigma)\)};
  \node[paneltitle] at (8.2,1.43) {\((F_A,\sigma')\)};
  \node[paneltitle] at (12.3,1.43) {\((G,\sigma')\)};

  \draw[flowarrow] (1.35,0)--(2.75,0)
    node[midway,above=3pt,flowlabel] {sample \(A\)};
  \draw[flowarrow] (5.45,0)--(6.85,0)
    node[midway,above=3pt,flowlabel] {flip}
    node[midway,below=3pt,flowlabel] {\(1\leftrightarrow2\)};
  \draw[flowarrow] (9.55,0)--(10.95,0)
    node[midway,above=3pt,flowlabel] {forget \(A\)};

  \draw[active] (1.3,-.6)--(1.63,-.6)
    node[right=2pt,legend] {active};
  \draw[inactive] (1.3,-.96)--(1.63,-.96)
    node[right=2pt,legend] {inactive};
\end{tikzpicture}
\caption{One soft flip update (\Cref{def:soft-kernel}), without pins.
Given \(\sigma\), each satisfied edge is activated independently with
probability \(1-x\); inactive (dashed) edges are absent from the hard
instance \(F_A\).  The flip kernel swaps the highlighted
\(\{1,2\}\)-component of \(F_A\), and discarding \(A\) returns the
updated colouring \(\sigma'\) to \(G\).}
\label{fig:soft-kernel-step}
\end{figure}
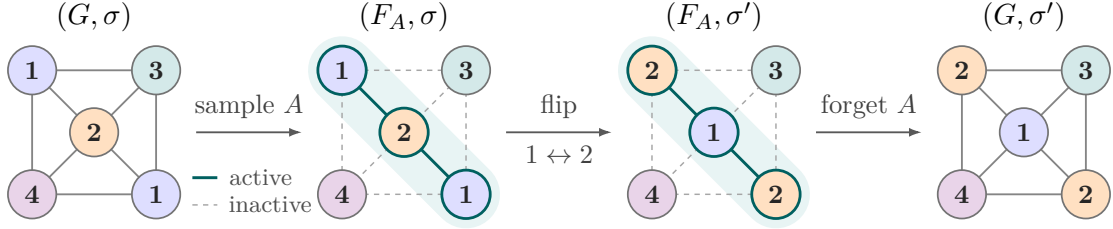

\begin{lemma}[Soft flip kernels]\label{lem:soft-stationary}
For every profile \(\rho\) and \(x\in(0,1]\),
\(K_{x,\rho}^{G,\tau}\) is reversible with stationary law
\(\mupin{G}{\tau}{x}\).  If \(\rho_1>0\), it is irreducible on
\(\colours^{V^\tau}\).
\end{lemma}

\begin{proof}
By \Cref{lem:active-rep}, conditional detailed balance for each \(A\),
multiplied by its marginal
probability in \eqref{eq:joint-active-law} and summed over \(A\), gives
detailed balance for the spin marginal.  The empty active set has positive
probability when \(x>0\); on this event, \(\rho_1>0\) permits every
one-site recolouring.
\end{proof}

\subsection{Common activations and boundary sensitivity}
\label{sec:boundary-sensitivity}

For colourings \(X,Y\) at Hamming distance one (adjacent colourings),
differing at a vertex \(v\), which we call the \emph{root}, with
\(X_v=a\ne b=Y_v\), draw one independent Bernoulli\((1-x)\) coin per
labelled constraint.  Activate that constraint on either side exactly when
the coin is one and the inequality is satisfied there.  Write
\(\mathcal A=(A_X,A_Y)\) for the outcome and \(F_X,F_Y\) for the hard
instances.

\begin{lemma}[The two hard instances differ only near \(v\)]
\label{lem:coupled-activation-root-local}
The common-coin construction has the correct activation marginals.
The two instances have the same induced graph and lists off \(v\).  Their
sets of edges at \(v\) can differ only at neighbours of colour \(a\) or
\(b\), and their lists at \(v\) can differ only in those two colours.
\end{lemma}

\begin{proof}
An inequality has different satisfaction status only if it is incident to
\(v\) and its other endpoint has colour \(a\) or \(b\).
\end{proof}

\begin{lemma}[Effect of adding one labelled free--pinned edge]
\label{lem:boundary-sensitivity}
Let a \emph{minus} and a \emph{plus} pinned system have the same free
vertex set \(W\), of size \(m\ge1\), and the same labelled constraints
except for one additional free--pinned edge in the plus system.  For a
profile \(\rho\) with \(R_\rho:=\sup_{s\ge1}s\rho_s<\infty\) and
\(x\in(0,1]\), their soft kernels \(K_-\) and \(K_+\) satisfy
\begin{equation}\label{eq:one-boundary-row}
  \WHam(K_-(\sigma,\cdot),K_+(\sigma,\cdot))
  \le\frac{(1-x)R_\rho}{mq}
  \qquad(\sigma\in\colours^W).
\end{equation}
For hard list kernels, deleting one colour from one list gives the bound
\(R_\rho/(mq)\) on colourings feasible for the smaller list.
Adding \(k\) labelled free--pinned edges (for hard list kernels,
deleting \(k\) colours from lists) multiplies the respective bound
by \(k\).
\end{lemma}

\begin{proof}
Couple the common activation coins and flip proposals.  Suppose the extra
constraint is the edge \(up\), which forbids the colour \(\tau_p\) at
\(u\).  It is inactive when \(\sigma_u=\tau_p\); otherwise it is active
with probability \(1-x\).  When active, it can change acceptance only for
the unique \(\{\sigma_u,\tau_p\}\)-component containing \(u\).  If this component has
size \(s\), couple its rejected plus move to its minus move: the cost is
at most \(s\rho_s/(mq)\le R_\rho/(mq)\).
Average over activation to prove \eqref{eq:one-boundary-row}; deterministic
list deletion omits the factor \(1-x\).  For several free--pinned edges
(or deleted colours), add them one at a time and use the triangle
inequality.
\end{proof}
\subsection{The conditional hard estimate}
\label{sec:hard-one-step}

Use Vigoda's profile~\cite{Vigoda2000}
\begin{equation}\label{eq:flip-params}
  p_1=1,\quad p_2=\frac{13}{42},\quad p_3=\frac16,\quad
  p_4=\frac2{21},\quad p_5=\frac1{21},\quad p_6=\frac1{84},\quad
  p_s=0\quad(s\ge7),
\end{equation}
and abbreviate \(\flip_F:=\Phi_F^{\mathbf p}\) and
\(\softK_x^{G,\tau}:=K_{x,\mathbf p}^{G,\tau}\).

Fix \(G\in\Gdeg\), a pinning \(\tau\), and \(x\in(0,1]\).
For colourings \(X,Y\) at Hamming distance one, with
\(X_v=a\ne b=Y_v\), condition on their common-coin activation outcome
\(\mathcal A\).  Write \(F_X,F_Y\) for the resulting hard instances,
\(L_{X,v}:=L_{A_X,v}\) and \(L_{Y,v}:=L_{A_Y,v}\) for their root lists,
\(n=|V^\tau|\), and \(G_0\) for their common induced graph off \(v\),
with its common colouring.  The two current colourings are feasible, so
\(a\in L_{X,v}\) and \(b\in L_{Y,v}\); the root lists agree on every
other colour.  Put
\[
  N_c=\{u\in N_{G^\tau}(v):X_u=Y_u=c,\;
    vu\text{ is active on at least one side}\},
  \qquad \ell=|L_{X,v}\cap L_{Y,v}|.
\]
The one-step estimate of Vigoda~\cite{Vigoda2000} and its list-colouring
extension in~\cite[Section~6 and Appendix~E]{CDMPP2019} couple two
colourings of one fixed instance that differ at a single vertex.  Here the
two chains run on the hard instances \(F_X\) and \(F_Y\), which may differ
in the edges at the root and in the root list
(\Cref{lem:coupled-activation-root-local}), so we prove the estimate
directly.

\begin{proposition}[Conditional hard estimate]\label{prop:hard-one-step}
The updates admit a coupling \((X',Y')\), with
\(X'\sim\flip_{F_X}(X,\cdot)\) and \(Y'\sim\flip_{F_Y}(Y,\cdot)\), such
that
\begin{equation}\label{eq:conditional-hard-estimate}
  nq\,\E[\Ham(X',Y')-1\mid\mathcal A]
  \le\frac{11}{6}\sum_c|N_c|-\ell.
\end{equation}
\end{proposition}

\begin{proof}
We first match component moves and then bound the Hamming drift
separately for each colour.  Detailed size estimates are needed only when
\(c\notin\{a,b\}\) and \(|N_c|\in\{1,2\}\); the other cases admit coarse
bounds.

Scale all move probabilities by \(nq\).  Write \(p(M)\) for the scaled
mass of a component move \(M\): it is \(p_{|M|}\) if the swap is feasible and
zero otherwise.  A root component contains \(v\); an off-root component
does not.  Match every move that is off-root in both instances with the
identical move in the other instance.

Call \(c\notin\{a,b\}\) a regular colour.  Every active edge from
\(v\) to \(N_c\) is then common to both instances.  Let \(\mathcal C\)
be the family of \(\{a,c\}\)-components of \(G_0\) meeting \(N_c\),
and let \(\mathcal D\) be their \(\{b,c\}\)-counterparts.  Adding
\(v\) joins each family into a root component on one side; on the other
side its members remain off-root.  The remaining moves, which form
group~\(c\), are:
\begin{center}
\begin{tabular}{@{}lll@{}}
\toprule
colour pair & moves in \(X\) & moves in \(Y\)\\
\midrule
\(\{a,c\}\) & root \(R\) & off-root \(C\in\mathcal C\)\\
\(\{b,c\}\) & off-root \(D\in\mathcal D\) & root \(S\)\\
\bottomrule
\end{tabular}
\end{center}
Here
\[
  R=\{v\}\cup\bigcup_{C\in\mathcal C}C,
  \qquad S=\{v\}\cup\bigcup_{D\in\mathcal D}D.
\]
A root swap is feasible exactly when its new root colour is allowed and
all its constituent off-root swaps are feasible: these are precisely its
vertexwise list tests.  This equivalence will also handle blocked moves
in the estimates below.

\paragraph{Matching the moves.}
Write \(m=|N_c|\).  If \(m=0\), match the two singleton root moves,
of mass one, when \(c\) is in both root lists; otherwise both are blocked.
If \(m\ge1\), choose largest components \(C_\star,D_\star\), breaking
ties by the least-indexed neighbour in a fixed ordering of \(N_c\).
Match the whole mass \(p(R)\) of \(R\) to \(C_\star\), and the whole
mass \(p(S)\) of \(S\) to \(D_\star\).  These allocations fit: since
\(|R|>|C_\star|\) and \(p_s\) is nonincreasing for \(s\ge1\), we have
\(p_{|R|}\le p_{|C_\star|}\), and if the swap of \(R\) is feasible, so is
the swap of its constituent \(C_\star\), by the equivalence above.  Hence
\(p(R)\le p(C_\star)\), and likewise \(p(S)\le p(D_\star)\).  The residual
off-root masses are
\[
  \tilde p(C)=p(C)-\one{C=C_\star}p(R),
  \qquad \tilde p(D)=p(D)-\one{D=D_\star}p(S).
\]
Write \(N_c=\{u_1,\ldots,u_m\}\) in the fixed ordering.  For each \(i\),
let \(C_i\in\mathcal C\) and \(D_i\in\mathcal D\) be the components
containing \(u_i\), and match
an additional mass
\begin{equation}\label{eq:incidence-matrix-entry}
  \min\left\{\frac{\tilde p(C_i)}{|C_i\cap N_c|},
                  \frac{\tilde p(D_i)}{|D_i\cap N_c|}\right\}
\end{equation}
between those two off-root moves.  Dividing by the number of incidences
ensures that no component's residual mass is exceeded.
\Cref{fig:incidence-allocation} illustrates this division when a single
component meets two root neighbours.

\begin{figure}[tbp]
\centering
\begin{minipage}[t]{0.34\textwidth}
\vspace{0pt}
\centering
\small\textbf{Shared off-root component}\par\medskip
\begin{tikzpicture}[
  x=1cm,y=1cm,
  vertex/.style={circle,draw=black!65,fill=white,line width=.65pt,
    minimum size=6.5mm,inner sep=1pt,font=\small},
  edge/.style={draw=black!65,line width=.75pt},
  clabel/.style={font=\small,text=blue!65!black},
  dlabel/.style={font=\small,text=orange!70!black}
]
  \coordinate (v) at (0,1.5);
  \coordinate (u1) at (-1.3,0);
  \coordinate (u2) at (1.3,0);
  \coordinate (w) at (0,-1.25);
  \draw[draw=blue!13,line width=9mm,line cap=round,line join=round]
    (u1)--(w)--(u2);
  \draw[edge] (v)--(u1)--(w)--(u2)--(v);
  \draw[orange!70!black,dashed,line width=.75pt]
    (u1) circle[radius=.43cm] (u2) circle[radius=.43cm];
  \node[vertex] at (v) {\(a/b\)};
  \node[vertex] at (u1) {\(c\)};
  \node[vertex] at (u2) {\(c\)};
  \node[vertex] at (w) {\(a\)};
  \node[above=4mm,font=\small] at (v) {\(v\)};
  \node[above left=3mm,font=\small] at (u1) {\(u_1\)};
  \node[above right=3mm,font=\small] at (u2) {\(u_2\)};
  \node[below=4mm,font=\small] at (w) {\(w\)};
  \node[clabel] at (0,-.45) {\(C\)};
  \node[dlabel] at (-1.55,-.65) {\(D_1\)};
  \node[dlabel] at (1.55,-.65) {\(D_2\)};
\end{tikzpicture}
\end{minipage}\hfill
\begin{minipage}[t]{0.63\textwidth}
\vspace{0pt}
\centering
\small\textbf{Root matches and incidence allocation}\par\medskip
\(R_X\longleftrightarrow C_Y:\ 8,\qquad
  (D_1)_X\longleftrightarrow S_Y:\ 14.\)\par\medskip
\setlength{\tabcolsep}{5pt}
\renewcommand{\arraystretch}{1.2}
\begin{tabular}{@{}lrrrrr@{}}
\toprule
Move & Initial & \shortstack{After root\\matching}
     & \shortstack{Used at\\\(u_1\)}
     & \shortstack{Used at\\\(u_2\)} & Left\\
\midrule
\(C_Y\)     & 14 &  6 & 3 & 3 &  0\\
\((D_1)_X\)& 84 & 70 & 3 & 0 & 67\\
\((D_2)_X\)& 84 & 84 & 0 & 3 & 81\\
\bottomrule
\end{tabular}
\[
  u_1:\ \min\{6/2,70/1\}=3,\qquad
  u_2:\ \min\{6/2,84/1\}=3.
\]
\end{minipage}
\caption{Incidence matching on a four-cycle with full lists, \(q=4\),
and all four edges active.  The root has colour \(a\) in \(X\) and \(b\)
in \(Y\).  Off the root,
\(C=\{u_1,w,u_2\}\) and \(D_i=\{u_i\}\), so \(|R|=4\) and \(|S|=3\);
choose \(D_\star=D_1\).  Subscripts \(X,Y\) indicate the instance in
which a move is made.
All numbers are \(84\) times the scaled masses and hence represent actual
probabilities in units of \(1/(84nq)\).  By \eqref{eq:flip-params},
\(84p_4=8\) for \(R\), \(84p_3=14\) for \(C\) and \(S\), and \(84p_1=84\)
for \(D_1\) and \(D_2\).  The root moves \(R_X\) and \(S_Y\) are matched in
full, so they leave no residual and have no row.
Each of the two ``Used at'' columns records both sides of the same match.
The two incidences use \(3+3=6\), exactly the residual mass of \(C\);
without the division by \(|C\cap N_c|=2\), each would claim \(6\), twice
what \(C\) has left.}
\label{fig:incidence-allocation}
\end{figure}

For the pair \(\{a,b\}\), the \(Y\)-root consists of \(v\) and the
\(\{a,b\}\)-components of \(G_0\) meeting \(N_a\); the \(X\)-root consists
of \(v\) and those meeting \(N_b\).  A component meeting both sets belongs
to both roots and gives no off-root move.  Put the \(Y\)-root move,
together with the \(X\)-off-root moves meeting \(N_a\) but not \(N_b\),
into group~\(a\), and define group~\(b\) symmetrically.  When \(|N_a|=1\) and its component \(C\)
is off-root on the \(X\)-side, match the whole \(Y\)-root mass to \(C\).
Otherwise match it to holding.  Do the symmetric construction for \(b\).
The match with \(C\) fits: the \(Y\)-root is then \(\{v\}\cup C\), so its
mass is at most \(p_{|C|+1}\le p_{|C|}\), and if the \(Y\)-root swap is
feasible, so is the swap of \(C\) in \(X\), since the vertices of \(C\)
have the same lists in both instances
(\Cref{lem:coupled-activation-root-local}).
Each side has at least \(n\) units of scaled holding mass, contributed by
the \(n\) proposals in which a vertex proposes its current colour.  On
each side, at most one root move uses the holding mass, and its scaled
mass is at most \(p_1=1\), so these allocations fit.

These groups exhaust the remaining side-labelled moves without overlap:
colour pairs avoiding \(a,b\) were common, and the preceding description
classifies all pairs containing either root colour.  Complete the partial
coupling by pairing the residual masses of \(X\)-moves and \(Y\)-moves
proportionally: if they are \(\omega_M\) and \(\omega'_{M'}\), with common
total \(T>0\), add mass \(\omega_M\omega'_{M'}/T\) to the pair \((M,M')\);
if \(T=0\), nothing remains.
All allocations have the required marginals after division by \(nq\).

\paragraph{The charges.}
A pair of moves \(M,M'\) coupled by this completion increases the
Hamming distance by at most \(|M|+|M'|\).  Charge each residual move its size times its residual mass;
let \(Q_c\) be these charges plus the actual drift of the prescribed
matches in group \(c\).  Common matches have zero drift, and the residual
cost separates by side even when the completion joins different groups.
Consequently, the first bound below holds, and it suffices to prove the
second for every colour \(c\):
\begin{align}
  nq\,\E[\Ham(X',Y')-1\mid\mathcal A]&\le\sum_cQ_c,
    \label{eq:cert-ham}\\
  Q_c&\le\frac{11}{6}|N_c|-\one{c\in L_{X,v}\cap L_{Y,v}}.
    \label{eq:cert-bound}
\end{align}

For a regular colour, matching \(R\) to \(C_\star\) creates
disagreements only on the other constituents of \(R\); its drift is
\(|R|-|C_\star|-1\).  The analogous cost for \(S\) is
\(|S|-|D_\star|-1\).  Each off-root match in
\eqref{eq:incidence-matrix-entry} saves at least one unit because
\(C_i\cap D_i\) contains \(u_i\).  Adding root-match costs and residual
size charges, and subtracting these matching savings, gives, for \(m\ge1\),
\begin{equation}\label{eq:regular-colour-charge}
\begin{aligned}
 Q_c\le{}&(|R|-|C_\star|-1)p(R)+(|S|-|D_\star|-1)p(S)\\
 &+\sum_C|C|\tilde p(C)+\sum_D|D|\tilde p(D)
 -\sum_{i=1}^m
   \min\left\{\frac{\tilde p(C_i)}{|C_i\cap N_c|},
                    \frac{\tilde p(D_i)}{|D_i\cap N_c|}\right\}.
\end{aligned}
\end{equation}

The following bounds follow directly from the six nonzero profile values:
\begin{equation}\label{eq:cert-ineq-1}
  rp_r\le1,\qquad
  rp_r\le\frac{13}{21}\ (r\ge2),\qquad
  (r-1)p_r\le\frac13,\qquad
  (r-2)p_r\le\frac4{21}.
\end{equation}
For two off-root moves meeting the same neighbour, with sizes \(r,s\)
and masses \(\theta,\theta'\), write
\(g_{r,s}(\theta,\theta')=r\theta+s\theta'-\min\{\theta,\theta'\}\) for their size charges minus the
guaranteed matching saving.  This function is nondecreasing in each
argument for \(r,s\ge1\).  Since \(p_r\) is nonincreasing,
\(g_{r,s}(p_r,p_s)=rp_r+(s-1)p_s\) when \(r\le s\).
Using \eqref{eq:cert-ineq-1}, with \(r=1\) or \(r\ge2\), gives
\begin{equation}\label{eq:cert-ineq-2}
  g_{r,s}(\theta/j,\theta'/k)\le\frac43
  \quad(0\le\theta\le p_r,\ 0\le\theta'\le p_s,\ j,k\ge1).
\end{equation}
Splitting each component charge among its incidences in \(N_c\) now
bounds the last three terms of \eqref{eq:regular-colour-charge} by
\(4m/3\).  If \(c\) is blocked at the root, both root masses vanish,
proving \eqref{eq:cert-bound}.  If it is allowed and \(m\ge3\), the
root terms are each at most \(4/21\), giving
\(Q_c\le8/21+4m/3\le11m/6-1\).  For \(m=0\), the singleton match
gives \(Q_c=-\one{c\in L_{X,v}\cap L_{Y,v}}\).

\paragraph{Regular colours with one or two neighbours.}
It remains to treat an allowed regular colour with \(m\in\{1,2\}\).
For \(m=1\), let the two off-root sizes be \(r,s\) and put
\(d_r=p_r-p_{r+1}\).  Root feasibility equals off-root feasibility,
so the residual masses lie in \(\{0,d_r\}\) and \(\{0,d_s\}\), respectively.
Thus \eqref{eq:regular-colour-charge} and monotonicity of \(g\) give
\(Q_c\le g_{r,s}(d_r,d_s)\).  The profile gives
\[
  d_1=\frac{29}{42},\qquad
  rd_r\le\frac27\ (r\ge2),\qquad
  (r-1)d_r\le\frac17\ (r\ge2).
\]
If both sizes are at least two, the bound is at most \(4/7\).
If both are one, it is \(29/42\); otherwise it is at most
\(29/42+1/7=5/6\), as required.

For \(m=2\), first suppose that some off-root component contains both
neighbours.  Its size is at least three, since those neighbours have the
same colour.  Its root term in \eqref{eq:regular-colour-charge} is zero,
and its off-root charge is at most \(rp_r\le1/2\) for \(r\ge3\).
The other family contributes at most two in off-root charges.  Its root
term is zero for one component, and at most \(k'p_{1+k+k'}\le1/6\) for
two sizes \(k\ge k'\).  Hence \(Q_c\le1/2+2+1/6=8/3\).

Otherwise each family has two distinct components.  Index them by their
neighbours, with sizes \(r_1,r_2\) and \(s_1,s_2\), and feasibility
bits \(f_i,f'_i\in\{0,1\}\).  Let \(P=p(R)\) and \(P'=p(S)\).  Before
allocating mass to the root matches, the off-root bound is
\[
  H_0=\sum_{i=1}^2 g_{r_i,s_i}(f_ip_{r_i},f'_ip_{s_i}).
\]
For \(\mathcal C\), the root match costs \((\min_i r_i)P\), removes
size charge \((\max_i r_i)P\), and loses at most \(P\) of matching
saving.  These three contributions explain its correction to \(H_0\).
Adding the corresponding correction for \(\mathcal D\) gives
\begin{equation}\label{eq:two-neighbour-reduction}
 Q_c\le H_0+(1-\max r_i+\min r_i)P
             +(1-\max s_i+\min s_i)P'.
\end{equation}
A correction is nonpositive for unequal sizes.  For two equal sizes
\(k\ge2\), it is at most \(p_{2k+1}\), which is nonzero only for
\(k=2\) and then equals \(1/21\).  For two singletons it is at most
\(p_3=1/6\).

Thus positive corrections occur only within equal-size families.
Singletons also have the largest off-root mass, whereas
\(rp_r\le13/21\) for every nonsingleton.  This suggests classifying
the four sizes by their number \(h\) of singletons.  A pair of
singletons has \(g_{1,1}\le1\); a pair of nonsingletons has
\(g_{r,s}\le rp_r+sp_s\le26/21\) by \eqref{eq:cert-ineq-1}; and every
pair satisfies \(g_{r,s}\le4/3\) by \eqref{eq:cert-ineq-2}.  These bounds
give the following upper bounds in \eqref{eq:two-neighbour-reduction}:
\begin{center}
\begin{tabular}{@{}clccc@{}}
\toprule
\(h\) & arrangement & \(H_0\) & total correction & \(Q_c\)\\
\midrule
0 & any & \(52/21\) & \(2/21\) & \(18/7\)\\
1 & any & \(4/3+26/21\) & \(1/21\) & \(55/21\)\\
2 & one singleton in each family & \(2\cdot4/3\) & \(0\) & \(8/3\)\\
3 & any & \(1+4/3\) & \(1/6\) & \(5/2\)\\
4 & any & \(2\) & \(2\cdot1/6\) & \(7/3\)\\
\bottomrule
\end{tabular}
\end{center}
Every last-column entry is at most \(8/3\).

It remains to handle \(h=2\) with both singletons in one family, say
\(\mathcal C\).  Here the general correction bound is too coarse.
If \(P>0\), both singleton swaps are feasible and their residual masses
are at least \(1-p_3=5/6>p_2\).  Since both opposite masses are at most
\(p_2\), subtracting \(P\) loses no matching saving.  The root cost
\(P\) and removed size charge \(P\) cancel, so \(\mathcal C\) has
zero correction.  The \(\mathcal D\)-correction is nonpositive for
unequal sizes and zero for equal sizes at least three.
Only \((s_1,s_2)=(2,2)\) remains, where
\[
  Q_c\le2g_{1,2}(p_1,p_2)+p_5
       =2(1+p_2)+p_5=\frac83;
\]
otherwise \(H_0\le8/3\) suffices.  This proves the regular-colour bound
entirely by the displayed scalar inequalities.

\paragraph{The two root colours.}
Consider colour \(a\), writing \(m=|N_a|\); colour \(b\) is symmetric.
If \(a\notin L_{Y,v}\), its root move is blocked and its at most \(m\)
off-root components each cost at most \(sp_s\le1\).
If \(a\in L_{Y,v}\) and \(m=0\), its singleton root/holding match
coalesces, giving \(Q_a=-1\).  For \(m\ge2\), let \(r\) be the size of
the \(Y\)-root.  Its root/holding match has drift \(r-2\) and mass at
most \(p_r\), contributing at most \((r-2)p_r\le4/21\); the at most \(m\)
off-root components again cost at most one each, so
\(Q_a\le m+4/21\le11m/6-1\).

In the remaining case \(a\in L_{Y,v}\) and \(m=1\), its root is
\(\{v\}\cup C\), where \(C\) has size \(s\).
If the root swap is infeasible, the corresponding off-root swap, when
present, fails the same list test, giving zero charge.  If the root swap
is feasible and \(C\) also meets \(N_b\), there is no off-root move in
this group and the charge is at most \((s-1)p_{s+1}\le4/21\).
Otherwise the root/off-root match coalesces, leaving charge at most
\[
  -p_{s+1}+s(p_s-p_{s+1})
   =sp_s-(s+1)p_{s+1}\le\frac8{21}<\frac56.
\]
This proves \eqref{eq:cert-bound} for all colours; summing proves the
proposition.
\end{proof}

\paragraph{Arithmetic check.}
\ifpdf
\embedfile[
  filespec=verifier.py,
  mimetype=text/x-python,
  desc={Exact-rational check of the finite inequalities in the proof of
    the conditional hard estimate}
]{\detokenize{verifier.py}}
\fi
The script \nolinkurl{verifier.py}, attached to this PDF and available as
an ancillary file with the arXiv source, independently checks the bounds
\(rp_r\le1\) and \((r-2)p_r\le4/21\), \eqref{eq:cert-ineq-2}, the two
one-neighbour root-colour bounds, and the maximum of the right-hand side
of \eqref{eq:regular-colour-charge} for \(m\le2\).  Run it with
\texttt{python3 verifier.py}.  Using exact rational arithmetic, it
enumerates component partitions, component sizes capped at seven, and
off-root feasibility bits, deriving root feasibility from the list-test
equivalence above; the analytic proof itself needs no enumeration.

\subsection{Averaging and stationary comparison}
\label{sec:activation-average}

\begin{lemma}[Soft path contraction]\label{lem:soft-contraction}
For \(G\in\Gdeg\), a pinning \(\tau\), and \(x\in(0,1]\), put
\(t=1-x\) and \(n=|V^\tau|\ge1\).  Colourings \(X,Y\) at Hamming
distance one admit a coupling \((X',Y')\) of one step of
\(\softK_x^{G,\tau}\) satisfying
\begin{equation}\label{eq:soft-averaged-drift}
  nq\,\E[\Ham(X',Y')-1]\le\frac{11}{6}t\Deg-q.
\end{equation}
In particular, if \(q>(11/6)t\Deg\), the soft kernel contracts Hamming
distance by the factor \(1-\kappa_x/n\), where
\begin{equation}\label{eq:soft-kappa}
  \kappa_x=\frac{q-(11/6)t\Deg}{q}>0.
\end{equation}
\end{lemma}

\begin{proof}
At the disagreement vertex \(v\), let \(d_v:=\deg_{G^\tau}(v)\) be its
free degree, \(d_\Lambda:=|N_G(v)\cap\Lambda|\) its number of pinned
neighbours, and \(k\) the number of free--pinned edges at \(v\) that are
active on at least one side.  Each constraint at \(v\) is satisfied on at
least one side, so it is active on at least one side exactly when its
Bernoulli\((t)\) coin is one.  Thus
\[
  \E\sum_c|N_c|=td_v,\qquad \E k=td_\Lambda,
  \qquad \ell\ge q-k.
\]
Averaging \Cref{prop:hard-one-step} gives, since
\(d_v+d_\Lambda=\deg_G(v)\le\Deg\),
\[
  nq\,\E[\Ham(X',Y')-1]
  \le\frac{11}{6}td_v+td_\Lambda-q
  \le\frac{11}{6}t\Deg-q.
\]
Divide by \(nq\); path coupling~\cite{BubleyDyer1997} along shortest
paths extends the adjacent-state bound to all pairs.
\end{proof}

\begin{proposition}[Vigoda-profile coupling for newly pinned laws]
\label{prop:full-interval-hci}
Let \(q,\Deg\) be positive integers with \(\Deg\ge2\).  For every
\(G\in\Gdeg\), every pinning \(\tau\), every \(r\in V^\tau\), all
\(a,b\in\colours\), and every \(x\in(0,1]\) satisfying
\(q>(11/6)(1-x)\Deg\),
\begin{equation}\label{eq:vigoda-master-ci}
  \WHam\!\left(
    \mupin{G}{\tau^{r=a}}{x},\mupin{G}{\tau^{r=b}}{x}
  \right)
  \le\frac{2(1-x)\Deg}{q-(11/6)(1-x)\Deg}.
\end{equation}
If \(q>11\Deg/6\), the bound \eqref{eq:vigoda-master-ci} also holds at
\(x=0\).
\end{proposition}

\begin{proof}
The cases \(a=b\) or \(m:=|V^\tau\setminus\{r\}|=0\) are immediate.
Otherwise compare each newly pinned system to the common middle system
\((G-r,\tau)\), whose Gibbs law \(\mu_{\rm mid}:=\mupin{G-r}{\tau}{x}\) is
stationary for \(\softK_x^{G-r,\tau}\) (\Cref{lem:soft-stationary}).
The child of colour \(c\) adds at most \(\Deg\) labelled free--pinned
edges to this middle system.  Since \(\sup_s sp_s=1\),
\Cref{lem:boundary-sensitivity} gives, for every common state \(\sigma\),
\begin{equation}\label{eq:soft-child-middle-row}
  \WHam\bigl(\softK_x^{G,\tau^{r=c}}(\sigma,\cdot),
              \softK_x^{G-r,\tau}(\sigma,\cdot)\bigr)
  \le\frac{(1-x)\Deg}{mq}.
\end{equation}
By \Cref{lem:soft-contraction}, applied to the child system with its
\(m\) free vertices, the child kernel contracts by the factor
\(1-\kappa_x/m\).
Apply \Cref{lem:stationary-comparison} with \(\kappa=\kappa_x/m\) and
\(C=(1-x)\Deg/(mq)\), using the stationary laws from
\Cref{lem:soft-stationary}, and then the triangle inequality:
\[
  \WHam(\mupin{G}{\tau^{r=c}}{x},\mu_{\rm mid})
  \le\frac{(1-x)\Deg}{q\kappa_x},
  \qquad
  \WHam(\mupin{G}{\tau^{r=a}}{x},\mupin{G}{\tau^{r=b}}{x})
  \le\frac{2(1-x)\Deg}{q\kappa_x}.
\]
This is \eqref{eq:vigoda-master-ci}.  If \(q>11\Deg/6\), both hard
pinned laws are well defined, with positive total weights, by
\Cref{lem:hard-feas}; their weights are polynomial in \(x\) and the
right-hand side of \eqref{eq:vigoda-master-ci} is continuous at \(x=0\),
so \Cref{lem:finite-endpoint-closure} passes the bound to \(x=0\).
The right-hand side of \eqref{eq:vigoda-master-ci} is largest at
\(x=0\), where it equals \(2/\gapq\); this proves \Cref{thm:strict-ci}.
\end{proof}

\begin{corollary}[Coupling input at \texorpdfstring{\(q=11\Deg/6\)}{q=11Delta/6}]
\label{cor:critical-line-input}
Let \(q,\Deg\) be integers with \(\Deg\ge2\) and \(q=11\Deg/6\).
Then \(\Gdeg\) satisfies coupling independence at \(x=0\) with a finite
constant, and for every \(\delta\in(0,1]\) it satisfies coupling
independence on \([\delta,1]\) with constant \(12/(11\delta)\).
\end{corollary}

\begin{proof}
For \(x\in[\delta,1]\), \Cref{prop:full-interval-hci} gives
\[
  \WHam\!\left(
    \mupin{G}{\tau^{r=a}}{x},
    \mupin{G}{\tau^{r=b}}{x}
  \right)
  \le \frac{12(1-x)}{11x}
  \le \frac{12}{11\delta}.
\]
At \(x=0\), integrality forces \(\Deg\) to be a positive multiple of \(6\),
so \cite[Theorem~20 and Proposition~22]{CFFGZZ2024} apply.
The conditional weights in \cite[Section~2]{CFFGZZ2024} omit pinned-only
edges and allow improper pinnings, agreeing with our normalization.
Apply the coupling bound of \cite[Theorem~20]{CFFGZZ2024} to
\(\tau^{r=a}\) and \(\tau^{r=b}\), then
project onto \(V^\tau\setminus\{r\}\), to obtain the bound required by
\Cref{def:potts-ci}.  As noted in~\cite{CFFGZZ2024}, this hard-colouring
bound is implicit in the coupling arguments
of~\cite{Liu2021Coupling,BlancaEtAl2022Coupling}.  Its list-contraction
input, \cite[Proposition~22]{CFFGZZ2024}, is derived there from the
list-colouring analysis of \cite[Section~6 and Appendix~E]{CDMPP2019},
which treats lists of equal size; we use \cite[Theorem~20]{CFFGZZ2024}
as stated.
\end{proof}

\subsection{Proof of the Potts zero-free theorem}
\label{sec:main-proof}

\begin{proof}[Proof of \Cref{thm:intro-main}]
The hypotheses imply \(q\ge\Deg+2\).  If \(q>11\Deg/6\), then
\Cref{thm:strict-ci} gives one coupling-independence constant on all of
\([0,1]\).  If \(q=11\Deg/6\), then \Cref{cor:critical-line-input}
supplies the two coupling inputs required by \Cref{thm:potts-transfer}.
Applying \Cref{thm:potts-transfer} to \(\Gdeg\) gives the normalized
zero-free statement.  The assertion for the unnormalized pinned
polynomial follows from \eqref{eq:intro-normalization}.
\end{proof}

\section{Zero-freeness for colour fields and Holant problems}
\label{sec:other-applications}

This section applies the separator induction of
\Cref{sec:transfer-principle} to two further families.
\Cref{sec:lee-yang} perturbs the colour fields of proper colourings
around the uniform field; only coupling independence at \(x=0\) is
needed.  \Cref{sec:holant-consequence} treats log-concave Boolean Holant
problems, where pinned instances carry residual signatures and the
coupling input comes from~\cite{ChenGu2023}.

\subsection{Colour fields}
\label{sec:lee-yang}

\subsubsection{Vertex-colour fields}

In the spirit of the Lee--Yang theorem~\cite{LeeYang1952}, we now
consider zeros in complex colour fields.
We use the same letter \(Z\) for both partition functions: an argument
\(z\) denotes the edge-activity polynomial, while an argument \(\lambda\)
denotes the field polynomial defined below.
For \(G=(V,E)\in\Gdeg\), a pinning \(\tau\colon\Lambda\to\colours\), and field
variables \(\lambda=(\lambda_{u,c})_{(u,c)\in V\times\colours}\),
define the normalized pinned field
polynomial
\begin{equation}\label{eq:field-Z}
  \nZpin{G}{\tau}(\lambda)
  :=\sum_{\varphi\colon V^\tau\to\colours}
     \prod_{u\in V^\tau}\lambda_{u,\varphi_u}
     \prod_{\substack{uv\in E\\\{u,v\}\not\subseteq\Lambda}}
       \one{(\tau\cup\varphi)_u\ne(\tau\cup\varphi)_v}.
\end{equation}
At the all-ones vector \(\lambda=\mathbf1\) on \(V\times\colours\), this
equals the hard-colouring value \(\nZpin{G}{\tau}(z)|_{z=0}\) of
\eqref{eq:potts-Z}, which is positive when \(q\ge\Deg+1\)
(\Cref{lem:hard-feas}).  The unpinned field
polynomial is \(\Zpart_G(\lambda)=\nZpin{G}{\varnothing}(\lambda)\).
The pinned field polynomial is
\[
  \Zpin{G}{\tau}(\lambda)
  :=\sum_{\substack{\sigma\colon V\to\colours\ \mathrm{proper}\\
                    \sigma|_\Lambda=\tau}}
       \prod_{v\in V}\lambda_{v,\sigma_v}.
\]
If \(\tau\) is proper on \(G[\Lambda]\), then
\begin{equation}\label{eq:full-field-normalization}
  \Zpin{G}{\tau}(\lambda)
  =\left(\prod_{v\in\Lambda}\lambda_{v,\tau_v}\right)
   \nZpin{G}{\tau}(\lambda|_{V^\tau\times\colours});
\end{equation}
otherwise \(\Zpin{G}{\tau}(\lambda)\equiv0\).

For \(v\in V^\tau\), conditioning on its colour gives
\begin{equation}\label{eq:field-recursion}
  \nZpin{G}{\tau}(\lambda)
  =\sum_{a:\bcount_v^\tau(a)=0}
     \lambda_{v,a}\,
     \nZpin{G}{\tau^{v=a}}
       (\lambda|_{(V^\tau\setminus\{v\})\times\colours}).
\end{equation}
The normalized child \(\nZpin{G}{\tau^{v=a}}\) remains defined when \(a\)
is blocked at \(v\);
such a colour simply does not occur in \eqref{eq:field-recursion}.

The following transfer needs coupling independence only at \(x=0\).

\begin{proposition}[Field transfer]
\label{prop:field-transfer}
Fix integers \(q,\Deg\) with \(\Deg\ge2\) and \(q\ge\Deg+1\), let
\(C_0\in[0,\infty)\), and let \(\mathcal G\subseteq\Gdeg\) be closed
under taking induced subgraphs and satisfy \(C_0\)-coupling independence
at \(x=0\).  There is \(\theta=\theta(q,\Deg,C_0)\in(0,1/2]\) such that,
for every \(G\in\mathcal G\) and every pinning \(\tau\),
\[
  \nZpin{G}{\tau}(\lambda)\ne0
  \qquad\text{whenever }|\lambda_{u,c}-1|<\theta
  \quad(u\in V^\tau,\ c\in\colours).
\]
\end{proposition}

\begin{proof}
We use the separator induction of \Cref{lem:potts-positive-response},
with base field \(\mathbf1\) and coupling constant \(C_0\).  Choose
\(L,s,N,\alpha\) as in that proof.  Up to relabelling, only finitely
many local polynomials occur below (bounded components and the interior
polynomials \(D_{c,\xi}\)), each in at most \((N+s)q\) field variables;
fix \(\theta\in(0,1/2]\) so small that the bounds on them used below
hold on the polydisc \(\max_{u,c}|\lambda_{u,c}-1|<\theta\).  Induct on
the number of free vertices, proving nonvanishing and the response bound
\begin{equation}\label{eq:field-response}
  \left|\Log
  \frac{\nZpin{G}{\tau^{v=a}}(\lambda)/
        \nZpin{G}{\tau^{v=a}}(\mathbf1)}
       {\nZpin{G}{\tau^{v=b}}(\lambda)/
        \nZpin{G}{\tau^{v=b}}(\mathbf1)}\right|\le\alpha.
\end{equation}
All normalized pinned base values are positive by \Cref{lem:hard-feas},
and logarithms continue from \(\lambda=\mathbf1\) along the segment to
\(\lambda\), which stays in the polydisc.  Bounded components are handled
by continuity over finitely many local instances.

For a larger component, choose a shell \(S\) with
\(\WHam(\nu_a,\nu_b)\le C_0/L\).  Include the interior and shell field
weights in the local polynomial \(D_{c,\xi}(\lambda)\), so that
\[
  \nZpin{G}{\tau^{v=c}}(\lambda)
  =\sum_\xi D_{c,\xi}(\lambda)\nZpin{G'}{\tau_\xi}(\lambda).
\]
Field perturbations preserve the hard constraints.  Thus each local
polynomial is either identically zero or has positive value at
\(\mathbf1\).  By the choice of \(\theta\), every nonzero local
polynomial satisfies
\[
  \eta_c(\xi):=\Log\frac{D_{c,\xi}(\lambda)}{D_{c,\xi}(\mathbf1)},
  \qquad |\eta_c(\xi)|\le\alpha/8.
\]
Set \(\eta_c(\xi)=0\) when \(D_{c,\xi}\equiv0\); such shell assignments
\(\xi\) have zero \(\nu_c\)-mass.  The exterior response
\[
  h_\xi:=\Log\frac{\nZpin{G'}{\tau_\xi}(\lambda)}
                    {\nZpin{G'}{\tau_\xi}(\mathbf1)}
\]
is \(\alpha\)-Lipschitz in Hamming distance by induction, applied as in
the proof of \Cref{lem:potts-positive-response} to the instance obtained
by unpinning one shell vertex, and has oscillation at most
\(s\alpha\le1/16\).  The separator identity becomes
\[
  R_c:=\frac{\nZpin{G}{\tau^{v=c}}(\lambda)}
             {\nZpin{G}{\tau^{v=c}}(\mathbf1)}
      =\E_{\nu_c}[e^{h+\eta_c}].
\]
The hypotheses of \Cref{lem:cavg}, with \(A=\alpha\) and
\(\gamma=\alpha/8\), hold along the segment from \(\mathbf1\) to
\(\lambda\), so, as in the proof of \Cref{lem:potts-positive-response},
the branch in \Cref{lem:cavg} is the logarithm continued from
\(\mathbf1\).  The lemma gives
\(|\Log(R_a/R_b)|\le2\alpha C_0/L+\alpha/2<\alpha\), proving
\eqref{eq:field-response}.

Finally, take a colour \(a_0\in\hardlist{\tau}{v}\), which exists by
\Cref{lem:hard-feas}.  After division by \(R_{a_0}\), the parent
recursion is
\[
  \frac{\nZpin{G}{\tau}(\lambda)}{R_{a_0}}
  =\sum_{a:\bcount_v^\tau(a)=0}
      \nZpin{G}{\tau^{v=a}}(\mathbf1)
      \lambda_{v,a}\frac{R_a}{R_{a_0}}.
\]
Each summand has positive real part: its first factor is positive,
\(|\arg\lambda_{v,a}|<\pi/6\) because \(|\lambda_{v,a}-1|<1/2\), and
\(|\Im\Log(R_a/R_{a_0})|\le\alpha\), so its argument has absolute value
less than \(\pi/6+\alpha<\pi/2\).  This closes the induction.
\end{proof}

\begin{theorem}[Lee--Yang polydisc]
\label{thm:lee-yang}
Let \(q,\Deg\) be integers with \(\Deg\ge2\).  Each of the following
parameter assumptions yields the conclusion below:
\begin{enumerate}[label=\textup{(\roman*)},leftmargin=2.2em]
  \item \(q\ge(11/6-1/84000)\Deg\);
  \item \(\Deg\ge125\) and \(q\ge1.809\Deg\);
  \item \(\Deg\ge3\) and \(q\ge\Deg+3\).
\end{enumerate}
In cases~\textup{(i)} and~\textup{(ii)}, put \(\mathcal G=\Gdeg\).
In case~\textup{(iii)}, there is an integer
\(g_{\rm LY}=g_{\rm LY}(q,\Deg)\ge3\) for which we put
\(\mathcal G=\Gdegg{g_{\rm LY}}\).  In each case there is
\(\theta=\theta(q,\Deg)\in(0,1/2]\) such that, for every
\(G\in\mathcal G\) and every pinning \(\tau\),
\[
  \nZpin{G}{\tau}(\lambda)\ne0
  \qquad\text{whenever }|\lambda_{u,c}-1|<\theta
  \quad(u\in V^\tau,\ c\in\colours).
\]
For a proper pinning, the corresponding statement for
\(\Zpin{G}{\tau}(\lambda)\) follows from
\eqref{eq:full-field-normalization} provided
\(\lambda_{v,\tau_v}\ne0\) for every \(v\in\Lambda\).  In particular, it
holds when the same polydisc condition is imposed at every vertex and
colour.
\end{theorem}

\begin{proof}
Each case implies \(q\ge\Deg+1\), and in each case \(\mathcal G\) is
closed under taking induced subgraphs.  By \Cref{prop:field-transfer},
it suffices to show that \(\mathcal G\) satisfies coupling independence
at \(x=0\) with a finite constant \(C_0=C_0(q,\Deg)\).  At \(x=0\), the
two laws in \Cref{def:potts-ci} are uniform distributions on list
colourings: in the instance on \(G^\tau\) in which \(r\) has list \(\colours\) and every
other free vertex \(u\) has list \(\hardlist{\tau}{u}\), pin \(r\) to
\(a\) and to \(b\).  This is the convention
of~\cite[Section~2]{CFFGZZ2024}, which allows improper pinnings (see the
proof of \Cref{cor:critical-line-input}).  As in the proof of
\Cref{lem:hard-feas}, every vertex \(u\) of this instance has a list of
size at least \(\deg_{G^\tau}(u)+q-\Deg\); in particular every colour
of \(\colours\) is feasible at \(r\).

In case~\textup{(ii)}, the input is \cite[Theorem~20]{CFFGZZ2024},
based on the flip contraction of~\cite{CarlsonVigoda2024}; projecting
its coupling onto \(V^\tau\setminus\{r\}\) gives the bound
\eqref{eq:potts-ci} at \(x=0\).  Case~\textup{(i)} reduces to the other inputs by
integrality.  If \(2\le\Deg<125\), then
\(6q\ge11\Deg-\Deg/14000>11\Deg-1\); since \(6q\) and \(11\Deg\) are
integers, this implies \(q\ge11\Deg/6\).  Apply \Cref{thm:strict-ci} at
\(x=0\) when the inequality is strict, and
\Cref{cor:critical-line-input} when equality holds.  If \(\Deg\ge125\),
then \(q\ge(11/6-1/84000)\Deg>1.809\Deg\), so case~\textup{(ii)}
applies.  The constant \(1/84000\) is the value of
\(\eps_0\) fixed in the proof of~\cite[Theorem~5.1]{CDMPP2019}.

In case~\textup{(iii)}, \cite[Theorem~5.10 and Lemmas~5.13
and~5.14]{CLMM2023} provide an integer \(g=g(q,\Deg)\) and an
influence-decay radius \(R=R(q,\Deg)\) with the
following property: for list colourings of a graph of maximum degree at
most \(\Deg\) and girth at least \(g\), with lists of size at most \(q\)
and at least the degree plus \(3\), pinning one vertex to two different
colours of its list changes the law of the remaining vertices by at most
\(2\Deg^{R}\) in the distance \(\WHam\).  This is recorded as
coupling independence in~\cite[Section~4]{CFFGZZ2024}.  Put
\(g_{\rm LY}:=\max\{g,3\}\).  For \(G\in\Gdegg{g_{\rm LY}}\), the graph
\(G^\tau\) has maximum degree at most \(\Deg\) and girth at least
\(g\), and, since \(q-\Deg\ge3\), the instance above satisfies these
hypotheses.  Hence \(C_0:=2\Deg^{R}\) works.
\end{proof}

\begin{remark}[Related results]
\label{rem:lee-yang-related}
For colour-class fields \(\lambda_{u,c}=\lambda_c\), Kuchukova, Perkins,
and Povill showed that \(\Zpart_G(\lambda)\ne0\) whenever \(q\ge2\Deg\)
and \(\max_c|\lambda_c-1|\) is at most an explicit radius of order
\(\Deg^{-4}\)~\cite[Theorem~1.14 and
Corollary~1.15]{KuchukovaPerkinsPovill2026}.  \Cref{thm:lee-yang} allows
vertex-dependent fields and applies below \(2\Deg\), but its radius is
not explicit.  For hypergraph colourings in the local lemma regime, Liu
and Yu obtained a Lee--Yang zero-free strip around
\([0,1]\)~\cite{LiuYu2026}.
\end{remark}

\subsubsection{Edge-colour fields}

For a finite simple graph \(G=(V,E)\), let \(L(G)\) be its line graph,
whose vertex set is \(E\).  Proper edge colourings of \(G\) are proper
vertex colourings of \(L(G)\).  An edge-colour pinning is a map
\(\tau\colon\Lambda_E\to\colours\) for some \(\Lambda_E\subseteq E\); it is
proper when adjacent pinned edges have distinct colours.  With field
variables \(\lambda=(\lambda_{e,c})_{e\in E,c\in\colours}\), define
\begin{equation}\label{eq:edge-field-Z}
  \widehat Z_G^{\mathrm{edge},\tau}(\lambda)
  :=\nZpin{L(G)}{\tau}(\lambda),
  \qquad
  Z_G^{\mathrm{edge},\tau}(\lambda)
  :=\Zpin{L(G)}{\tau}(\lambda).
\end{equation}
We write \(Z_G^{\mathrm{edge}}:=Z_G^{\mathrm{edge},\varnothing}\).  If
\(\tau\) is proper, then
\begin{equation}\label{eq:edge-full-field-normalization}
  Z_G^{\mathrm{edge},\tau}(\lambda)
  =\left(\prod_{e\in\Lambda_E}\lambda_{e,\tau_e}\right)
    \widehat Z_G^{\mathrm{edge},\tau}(\lambda);
\end{equation}
otherwise \(Z_G^{\mathrm{edge},\tau}(\lambda)\equiv0\).

\begin{corollary}[Edge-colour Lee--Yang polydisc]
\label{cor:edge-lee-yang}
Let \(q,\Deg\) be integers with \(\Deg\ge2\) and \(q\ge3\Deg\).
There exists \(\theta=\theta(q,\Deg)\in(0,1/2]\) such that, for every
finite simple graph \(G=(V,E)\) of maximum degree at most \(\Deg\) and
every edge-colour pinning \(\tau\),
\[
  \widehat Z_G^{\mathrm{edge},\tau}(\lambda)\ne0
  \qquad\text{whenever}\qquad
  \abs{\lambda_{e,c}-1}<\theta
  \quad(e\in E\setminus\Lambda_E,\ c\in\colours).
\]
In particular, \(Z_G^{\mathrm{edge}}(\lambda)\ne0\) in the same
polydisc.  For a proper pinning,
\(Z_G^{\mathrm{edge},\tau}(\lambda)\ne0\) when the same condition is
imposed at every edge-colour pair.
\end{corollary}

\begin{proof}
Put \(\Deg':=2\Deg-2\) and let \(\mathcal H_\Deg\) be the class of line
graphs \(L(G)\) with \(\Delta(G)\le\Deg\).  Every graph in
\(\mathcal H_\Deg\) has maximum degree at most \(\Deg'\), the class is
closed under taking induced subgraphs, and \(q\ge3\Deg\ge\Deg'+1\).  By
\Cref{prop:field-transfer}, applied with degree bound \(\Deg'\), it
suffices to show that \(\mathcal H_\Deg\) satisfies coupling
independence at \(x=0\) with constant \(C_0=2\Deg\).

Fix an edge-colour pinning \(\tau\) of \(G\), a free edge \(r\), and
colours \(a,b\).  As in the proof of \Cref{thm:lee-yang}, the two laws in
\Cref{def:potts-ci} arise from the list-edge-colouring instance on the
free edges of \(G\) in which \(r\) has list \(\colours\) and every other
free edge \(e\) has list \(\colours\) minus the colours of its pinned
neighbours, by pinning \(r\) to \(a\) and to \(b\).  Let \(\deg(e)\) be
the number of free edges adjacent to \(e\).  At most \(2\Deg-2\) edges
of \(G\) are adjacent to \(e\), and each pinned one removes at most one
colour from the list of \(e\), so every list, including that of \(r\),
has size at least \(\deg(e)+q-(2\Deg-2)\ge\deg(e)+\Deg+2\).  Hence the
instance is \(((1+\gamma)\Deg+1)\)-extra in the sense
of~\cite{CWZZ2025}, with \(\gamma=1/\Deg\), and
\cite[Theorem~15]{CWZZ2025} bounds \(\WHam\) between the two pinned laws,
on all free edges, by \(1+2/\gamma=1+2\Deg\).
Every coupling of the two laws disagrees at \(r\) when \(a\ne b\), so the
distance on the free edges other than \(r\) is at most \(2\Deg\).  The
conclusion for \(Z_G^{\mathrm{edge},\tau}\) follows from
\eqref{eq:edge-full-field-normalization}.
\end{proof}

\subsection{Log-concave Boolean Holant problems}
\label{sec:holant-consequence}

We next apply the method to log-concave Boolean Holant problems, using only
the notation needed here.  A symmetric Boolean signature of arity \(d\) is specified by
\[
  f\colon\{0,\ldots,d\}\longrightarrow[0,\infty),
\]
where \(f(k)\) is the value assigned when exactly \(k\) incident edges are
selected.  Such a signature is \emph{log-concave} if its positive support
is an interval and
\[
  f(k)^2\ge f(k-1)f(k+1)
  \qquad(1\le k<d).
\]
For a graph \(G=(V,E)\), signatures
\(\mathbf f=(f_v)_{v\in V}\) with
\(\operatorname{arity}(f_v)=\deg_G(v)\), and edge activities
\(\mathbf z=(z_e)_{e\in E}\), put
\begin{equation}
  Z_{G,\mathbf f}(\mathbf z)
  :=\sum_{S\subseteq E}
       \prod_{v\in V}f_v\bigl(\deg_S(v)\bigr)
       \prod_{e\in S}z_e ,
  \label{eq:holant-short-Z}
\end{equation}
where \(\deg_S(v)\) counts the selected edges incident to \(v\).
Zero locations and approximation algorithms for Holant polynomials were
studied in \cite{GLLZ2019,Casel2019}.  For log-concave signatures satisfying
\(f_v(0)>0\), Chen and Gu established spectral independence and rapid
sampling through recursive couplings~\cite{ChenGu2023}, while
He, Li, Qiu, and Zhang obtained an FPTAS~\cite{HLQZ2025}.
Chen, Chen, and Zhang proved a spectral-independence bound for these
measures that does not depend on the maximum degree~\cite{ChenChenZhang2026};
our argument uses the degree-dependent coupling input
of~\cite{ChenGu2023}, and \Cref{rem:matching-degree-dependence} shows
that complex zero-free widths must depend on \(\Deg\).

\subsubsection{Residual signatures and coupling independence}

\begin{lemma}[Residual Holant input]
\label{lem:holant-residual-input}
Fix an integer \(\Deg\ge2\), a finite family \(\mathcal F\) of log-concave
symmetric Boolean signatures satisfying \(f(0)>0\) for every
\(f\in\mathcal F\), and \(R>0\).  Let \(\mathcal F_{\rm res}\) consist
of all normalized residual signatures obtained as follows: for every
\(f\in\mathcal F\) and integers \(j,r\ge0\) with
\(j+r\le\operatorname{arity}(f)\) and \(f(j)>0\), take
\[
  g(k):=\frac{f(j+k)}{f(j)}
  \qquad(0\le k\le r)
\]
in \(\mathcal F_{\rm res}\).  Here \(j\) is the number of incident edges
fixed to \(1\), and \(r\) is the number left free; the remaining incident
edges are fixed to \(0\).  This is a finite family of log-concave
signatures, each with initial-interval support and \(g(0)=1\).

There is \(C=C(\Deg,\mathcal F,R)<\infty\) with the following property.
Let \(H=(W,F)\) be a finite simple graph of maximum degree at most
\(\Deg\), carrying at each vertex \(w\) a signature
\(g_w\in\mathcal F_{\rm res}\) of arity \(\deg_H(w)\), and let
\(\mathbf x\in[0,R]^F\).  For an edge \(e=uv\in F\), its normalized
\(0\)-child is obtained by deleting \(e\) and restricting the signature
of each endpoint \(w\in\{u,v\}\) to its first \(\deg_{H-e}(w)+1\) entries.  If
\(g_u(1)g_v(1)>0\), its normalized \(1\)-child instead has endpoint
signatures
\[
  k\longmapsto\frac{g_w(k+1)}{g_w(1)}
  \qquad(0\le k\le\deg_{H-e}(w)),\quad w\in\{u,v\}.
\]
All other signatures are unchanged, and every child signature belongs
to \(\mathcal F_{\rm res}\).  The normalization removes the fixed-edge
factor \(z_e g_u(1)g_v(1)\), so the \(1\)-child remains defined when
\(x_e=0\).  If \(g_u(1)g_v(1)>0\), let
\(\mu_{\mathbf x}^{e=0}\) and \(\mu_{\mathbf x}^{e=1}\) be the Gibbs laws
of the two normalized children on \(\{0,1\}^{F\setminus\{e\}}\).  Then
\[
  \WHam\bigl(\mu_{\mathbf x}^{e=0},
             \mu_{\mathbf x}^{e=1}\bigr)\le C.
\]
If \(Z^{e=0}\) and \(Z^{e=1}\) denote the corresponding normalized child
partition functions, then
\begin{equation}\label{eq:holant-child-monotonicity}
  Z^{e=1}(\mathbf x)\le Z^{e=0}(\mathbf x).
\end{equation}
When \(g_u(1)g_v(1)=0\), the unnormalized \(1\)-child is identically zero.
\end{lemma}

\begin{proof}
The claims about \(\mathcal F_{\rm res}\) follow because restriction and
multiplication by a positive scalar preserve log-concavity; the assumption
\(f(0)>0\) makes every resulting support an initial interval.  If
\(g(k)=f(j+k)/f(j)\) has parameters \((j,r)\), then its restriction to
\(\{0,\ldots,r'\}\) with \(r'\le r\) has parameters \((j,r')\), and for
\(g(t)>0\) the shift \(k\mapsto g(t+k)/g(t)=f(j+t+k)/f(j+t)\) has
parameters \((j+t,r-t)\).  Hence \(\mathcal F_{\rm res}\) is closed under
restriction and normalized shifts, so every child signature belongs to
\(\mathcal F_{\rm res}\).  Let
\[
  A:=\max\bigl(\{0\}\cup
       \{g(1):g\in\mathcal F_{\rm res},\
                  \operatorname{arity}(g)\ge1\}\bigr).
\]
When all free activities are positive, apply the two endpoint shifts
successively, shifting each original sequence before restricting it to
the remaining degree.  An isolated endpoint changes only a removed
scalar and requires no coupling.
By \cite[Proposition~17 and Lemma~21]{ChenGu2023}, applying the two
endpoint comparisons as in the proof of Proposition~14 there gives the
Wasserstein bound with
\[
  C:=2\bigl((1+A^2R)^\Deg-1\bigr).
\]
Here the intermediate shifted signatures remain in
\(\mathcal F_{\rm res}\), so each of positive arity has \(g(1)\le A\), and
every positive free activity is at most \(R\).  Thus each endpoint
comparison contributes at most \((1+A^2R)^\Deg-1\), and the triangle
inequality gives the displayed value of \(C\).  This bound does not
depend on a positive lower bound for the free activities.
If some activities vanish, replace only those
coordinates by positive values tending to zero.  Each normalized child
has the all-zero configuration with weight one, so both Gibbs laws
converge on the finite common state space.  Compactness of the coupling
polytope gives a limiting coupling with the same bound.

To prove \eqref{eq:holant-child-monotonicity}, note that the successive
ratios of a log-concave
signature \(g\) are nonincreasing on its positive support.  Since
\(g(0)=1\), multiplying these ratios gives
\begin{equation}\label{eq:holant-shift-monotonicity}
  \frac{g(t+k)}{g(t)}\le g(k)
  \qquad(g(t)>0,\quad t+k\le\operatorname{arity}(g)).
\end{equation}
If the numerator vanishes, the inequality is immediate.
Taking \(t=1\), each normalized \(1\)-child endpoint signature is
pointwise at most its \(0\)-child counterpart.  The termwise comparison over
subsets of \(F\setminus\{e\}\) proves
\eqref{eq:holant-child-monotonicity}.
\end{proof}

\subsubsection{Zero-freeness near the nonnegative orthant}

\begin{theorem}[Log-concave Holant polytube]
\label{thm:holant-box}
Fix an integer \(\Deg\ge2\) and a finite family \(\mathcal F\) of log-concave
symmetric Boolean signatures satisfying \(f(0)>0\) for every
\(f\in\mathcal F\).  For every \(R>0\), there is
\(\eps=\eps(\Deg,\mathcal F,R)>0\) such that the following holds.
For every finite simple graph \(G=(V,E)\) of maximum degree at most \(\Deg\),
every assignment \((f_v)_{v\in V}\) with \(f_v\in\mathcal F\) of arity
\(\deg_G(v)\), and every \(\mathbf z\in\C^E\),
\[
  Z_{G,\mathbf f}(\mathbf z)\ne0
  \qquad\text{whenever}\qquad
  \mathbf z\in\mathcal U_\eps([0,R])^E.
\]
The activities may be perturbed independently.  In particular, the union
\(\bigcup_{R>0}\mathcal U_{\eps(\Deg,\mathcal F,R)}([0,R])^E\) of these
polytubes, each with one \(R\) for all coordinates, is an open zero-free
neighbourhood of \([0,\infty)^E\).
On the diagonal \(z_e=z\), the univariate specialization
\(Z_{G,\mathbf f}(z\mathbf1)\) is nonzero throughout the open set
\[
  U_{\Deg,\mathcal F}
  :=\bigcup_{R>0}\mathcal U_{\eps(\Deg,\mathcal F,R)}([0,R])
  \supset[0,\infty),
\]
which depends only on \((\Deg,\mathcal F)\) and is common to all the
graphs and signature assignments above.
\end{theorem}

\begin{proof}
\textbf{Residual instances.}\hspace{1em}Let \(\mathcal F_{\rm res}\) be the residual family constructed in
\Cref{lem:holant-residual-input} for the present \(\Deg,\mathcal F\).
We prove the stronger assertion for every normalized instance on a graph
of maximum degree at most \(\Deg\) whose signatures belong to
\(\mathcal F_{\rm res}\).  Call an assignment of edge states
\emph{structurally feasible} if its signature product is nonzero, the
activities being treated as indeterminates.  This class of normalized
instances is closed under
fixing edges to structurally feasible states: removing the fixed
activity factors and normalizing each shifted signature by its value
at \(0\) again gives an instance with signatures in \(\mathcal F_{\rm res}\),
as in \Cref{lem:holant-residual-input}.
Dividing each original \(f_v\) by \(f_v(0)>0\) puts the instance in this
class and only multiplies its partition function by a positive constant.

\paragraph{Shell decomposition and exterior response.}
The interaction graph of the edge variables, namely the line graph of the
underlying graph, has maximum degree at most \(\Deg':=2\Deg-2\).  Let
\(C\) be the constant from \Cref{lem:holant-residual-input}, take an
integer \(L\ge\max\{1,64C\}\), and set
\[
  s:=\Deg'(\Deg'-1)^{L-1},
  \qquad
  N:=1+\Deg'\sum_{k=0}^{L-1}(\Deg'-1)^k.
\]
Choose \(0<\alpha<\pi/4\) with \(s\alpha\le1/16\).  Up to isomorphism
and apart from their activities, there are only finitely many connected
normalized instances with at most \(N\) edges.  Compactness over
\([0,R]^N\) therefore gives a common local radius on which their
partition functions are nonzero and the normalized child-response
quotients in \eqref{eq:holant-response} below have logarithms of modulus
at most \(\alpha\).

We now prove simultaneously, by strong induction on the number of free
edges, the following two assertions for every normalized instance
\(H=(W,F)\), every \(\mathbf x\in[0,R]^F\), and every \(\mathbf z\in\C^F\)
with \(\max_{e'\in F}|z_{e'}-x_{e'}|<\eps\): the partition function of
\(H\) is nonzero at \(\mathbf z\), and every edge \(e\in F\) whose
unnormalized \(1\)-child is not identically zero satisfies
\begin{equation}\label{eq:holant-response}
  \left|
  \Log
  \frac{Z^{e=1}(\mathbf z)/Z^{e=1}(\mathbf x)}
       {Z^{e=0}(\mathbf z)/Z^{e=0}(\mathbf x)}
  \right|\le\alpha.
\end{equation}
The logarithm is analytically continued from
\(\mathbf z=\mathbf x\).  Both assertions are immediate when there is no
free edge.

Fix an instance \(H=(W,F)\) with \(n:=|F|\ge1\), and fix a root edge
\(e\in F\) whose \(1\)-child survives, that is, whose unnormalized
\(1\)-child is not identically zero.  We first prove the response bound.
Components of the interaction graph not containing \(e\) give
the same factors in the two children; the induction hypothesis makes
them nonzero, so they cancel.  For the response bound we may therefore
replace \(H\) by the
component containing \(e\), relabel it as \(H=(W,F)\), reset
\(n:=|F|\), and write \(Z\) for its partition function.  Writing
\(\Gamma\) for its interaction graph, let
\[
  S_j(e):=\{e'\in F:\dist_{\Gamma}(e,e')=j\},
  \qquad 1\le j\le L.
\]
If some \(S_j(e)\) is empty, this component has at most \(N\) edges and
the local compactness choice applies.  Otherwise, project a coupling of
the two real child laws onto these \(L\) disjoint shells.  Their projected
costs sum to at most \(C\), so one shell \(S=S_{j_*}(e)\) has cost at
most \(C/L\); also \(|S|\le s\).  Put
\[
  I:=\{e'\in F\setminus\{e\}:1\le\dist_{\Gamma}(e,e')<j_*\},
  \qquad
  O:=\{e'\in F:\dist_{\Gamma}(e,e')>j_*\}.
\]

For \(i\in\{0,1\}\), let \(\Xi_i\subseteq\{0,1\}^S\) be the shell
assignments admitting a structurally feasible completion in the
\(i\)-child.
Initial-interval support makes each \(\Xi_i\) downward closed and
ensures that it contains \(\mathbf0\).  The termwise comparison in
\Cref{lem:holant-residual-input} also gives \(\Xi_1\subseteq\Xi_0\);
we use the common domain \(\Xi:=\Xi_0\).

To identify the exterior explicitly, let \(W_O\) be the vertices incident
to \(O\), and let \(t_w(\xi)\) count selected shell edges at \(w\).
No vertex in \(W_O\) is incident to \(e\) or \(I\), by the definition
of the line-graph shells.  For \(\xi\in\Xi\), the normalized exterior
signatures are therefore
\[
  g_{w,\xi}(k)=\frac{g_w(t_w(\xi)+k)}{g_w(t_w(\xi))}
  \qquad(w\in W_O,\quad 0\le k\le\deg_O(w)).
\]
Structural feasibility and initial-interval support make every denominator
positive.  These signatures belong to \(\mathcal F_{\rm res}\) and are
independent of the root state.  Write \(Z_\xi\) for their partition
function on \(O\).  Sum over \(I\), absorbing the shell monomial and
the removed factors \(g_w(t_w(\xi))\) into \(D_{i,\xi}\), to obtain
\[
  Z^{e=i}(\mathbf z)
  =\sum_{\xi\in\Xi_i}D_{i,\xi}(\mathbf z)Z_\xi(\mathbf z),
  \qquad i\in\{0,1\}.
\]
Structurally infeasible
assignments are absent from the sum; an assignment selecting an edge with
zero base activity remains in
the sum and merely has zero real weight.

By induction, \(Z_\xi(\mathbf z)\ne0\), so define
\[
  h_\xi:=\Log\frac{Z_\xi(\mathbf z)}{Z_\xi(\mathbf x)}
  \qquad(\xi\in\Xi).
\]
Write \(h(\xi):=h_\xi\).  If
\(\xi,\xi'\in\Xi\) differ in one coordinate, fix the root and every
edge of \(I\) to zero, and
all other shell edges to their common states.  Leaving only the changed
shell edge and the exterior free gives a normalized instance with
\(|O|+1<n\) free edges.  Both assignments remain structurally feasible
when the interior is set to zero, by initial-interval support; the two
normalized children of this instance at the changed shell edge are
exactly \(Z_\xi,Z_{\xi'}\).
Both logarithm branches vanish at \(\mathbf z=\mathbf x\), so uniqueness
of analytic continuation identifies \(h_\xi-h_{\xi'}\) with the
inductive response logarithm.  Any two elements \(\xi,\xi'\in\Xi\) can be
joined through their coordinatewise minimum \(\xi\wedge\xi'\) by a Hamming path contained in
\(\Xi\).  Hence
\[
  |h_\xi-h_{\xi'}|\le\alpha\Ham(\xi,\xi'),
  \qquad
  \max_{\xi,\xi'\in\Xi}|h_\xi-h_{\xi'}|
     \le s\alpha\le\frac1{16}.
\]

\paragraph{Error control.}
Extend the real shell marginals by zero from \(\Xi_i\) to \(\Xi\):
\[
  \nu_i(\xi):=
  \frac{D_{i,\xi}(\mathbf x)Z_\xi(\mathbf x)}
       {Z^{e=i}(\mathbf x)}
  \quad(\xi\in\Xi_i),
  \qquad
  \nu_i(\xi):=0\quad(\xi\notin\Xi_i).
\]
Then \(\WHam(\nu_0,\nu_1)\le C/L\).  Put
\(M_i:=\E_{\nu_i}[e^h]\).  Adding and subtracting the local factors at
\(\mathbf x\) gives
\[
  \frac{Z^{e=i}(\mathbf z)}{Z^{e=i}(\mathbf x)}=M_i+E_i,
\]
\[
  E_i=
  \frac{\displaystyle\sum_{\xi\in\Xi_i}
    \bigl(D_{i,\xi}(\mathbf z)-D_{i,\xi}(\mathbf x)\bigr)
    Z_\xi(\mathbf x)e^{h_\xi}}
       {Z^{e=i}(\mathbf x)}.
\]
Each \(D_{i,\xi}\) involves at most \(|I\cup S|\le N-1\) variables.
Only vertices touching \(S\) contribute nontrivial exterior normalization
constants, so there are at most \(2s\) such factors.  All local data
thus range over a finite family.  After requiring \(\eps\le1\), there is
a graph-independent \(K\ge1\) such that
\[
  |D_{i,\xi}(\mathbf z)-D_{i,\xi}(\mathbf x)|\le K\eps
  \qquad(i\in\{0,1\},\ \xi\in\Xi_i).
\]
The all-zero interior assignment gives \(D_{i,\mathbf0}(\mathbf x)\ge1\), whence
\(Z^{e=i}(\mathbf x)\ge Z_{\mathbf0}(\mathbf x)\).  Repeated pointwise
use of \eqref{eq:holant-shift-monotonicity} on the exterior gives
\(Z_\xi(\mathbf x)\le Z_{\mathbf0}(\mathbf x)\).  Enlarging \(K\) to
absorb \(|\Xi|\) and \(e^{1/16}\), we obtain
\[
  |E_i|\le K\eps|e^{h_{\mathbf0}}|,
  \qquad
  \left|\frac{M_i}{e^{h_{\mathbf0}}}-1\right|
       \le e^{1/16}-1,
  \qquad
  |M_i|\ge\frac34|e^{h_{\mathbf0}}|.
\]
These bounds use only structural feasibility, not positivity of the
base activities, and are therefore uniform on the closed box
\([0,R]^F\).

Apply \Cref{lem:cavg} on \(\Xi\), with \(A=\alpha\),
\(d=\Ham\), and no local perturbation.  It gives
\[
  \left|\Log\frac{M_1}{M_0}\right|
  \le\frac{2\alpha C}{L}\le\frac{\alpha}{32}.
\]
The preceding disc estimate fixes the same logarithm branch as
continuation from \(\mathbf z=\mathbf x\).  For
\(\zeta_i:=E_i/M_i\), we have \(|\zeta_i|\le4K\eps/3\).  Require
\(\eps<\min\{3/(8K),3\alpha/(32K)\}\).  Then \(|\zeta_i|<1/2\), and
\(|\Log(1+\zeta_i)|\le2|\zeta_i|\) bounds the two corrections in
\[
  \Log\frac{M_1+E_1}{M_0+E_0}
  =\Log\frac{M_1}{M_0}+\Log(1+\zeta_1)-\Log(1+\zeta_0)
\]
by a total of at most \(16K\eps/3<\alpha/2\).  This proves
\eqref{eq:holant-response} with room to spare.

\paragraph{Parent nonvanishing.}
Having proved the response bound for every surviving child pair, choose
any edge \(e=uv\) to prove nonvanishing for the current instance, and put
\(a_e:=g_u(1)g_v(1)\).  If \(a_e=0\), the unnormalized \(1\)-child is
identically zero and \(Z(\mathbf z)=Z^{e=0}(\mathbf z)\ne0\) by
induction.  Otherwise,
\[
  Z(\mathbf z)=Z^{e=0}(\mathbf z)
       +a_ez_eZ^{e=1}(\mathbf z).
\]
The possible positive values of \(a_e\) form a finite set, and
\[
  r:=\frac{Z^{e=1}(\mathbf x)}{Z^{e=0}(\mathbf x)}\le1
\]
by \eqref{eq:holant-child-monotonicity}.  With \(\ell\) denoting the
logarithm in \eqref{eq:holant-response},
\[
  \frac{Z(\mathbf z)}{Z^{e=0}(\mathbf z)}
  =1+a_e x_e r e^\ell+a_e(z_e-x_e)r e^\ell.
\]
The sum of the first two terms has real part at least one, whereas the
last term has modulus at most \(a_e e^\alpha\eps\).  A final uniform reduction of
\(\eps\) makes this modulus smaller than one and closes the induction.

For a point in the stated polytube, choose
\(x_e\in[0,R]\) with \(|z_e-x_e|<\eps\) for every edge and apply the
induction.  Taking the union over \(R>0\) gives the asserted open
neighbourhood of \([0,\infty)^E\).
On the diagonal, a single such \(R\) works for all coordinates, giving
the asserted univariate domain.
\end{proof}

\subsubsection{\texorpdfstring{\(b\)-matchings and \(b\)-edge covers}
{b-matchings and b-edge covers}}

\begin{corollary}[\(b\)-matchings]
\label{cor:bmatching-short}
Fix an integer \(\Deg\ge2\) and \(R>0\).  There is
\(\eps=\eps(\Deg,R)>0\) such that, for every finite simple graph \(G=(V,E)\) of
maximum degree at most \(\Deg\), every integer vector
\(\mathbf b=(b_v)_{v\in V}\) with \(0\le b_v\le\deg_G(v)\), and every
\(\mathbf z\in\C^E\),
\[
  Z^{\mathrm{match}}_{G,\mathbf b}(\mathbf z)
  :=\sum_{\substack{S\subseteq E\\
                    \deg_S(v)\le b_v\ \forall v}}
       \prod_{e\in S}z_e
  \ne0
\]
whenever \(\mathbf z\in\mathcal U_\eps([0,R])^E\).
Thus \(\bigcup_{R>0}\mathcal U_{\eps(\Deg,R)}([0,R])^E\)
is an open zero-free neighbourhood of \([0,\infty)^E\).
\end{corollary}

\begin{proof}
Apply \Cref{thm:holant-box} to the finite family
\(f_{d,b}(k)=\one{k\le b}\) for integers
\(0\le b\le d\le\Deg\) and \(0\le k\le d\).
\end{proof}

\begin{remark}[Dependence on the maximum degree]
\label{rem:matching-degree-dependence}
For ordinary matchings on the star \(K_{1,\Deg}\), the diagonal
partition function is \(1+\Deg z\).  Its zero at \(-1/\Deg\) shows
that the width in \Cref{cor:bmatching-short} must satisfy
\(\eps(\Deg,R)\le1/\Deg\).  Thus dependence on the maximum degree
is unavoidable for these neighbourhoods of \([0,R]\).  For ordinary
matchings much more is known: by the Heilmann--Lieb
theorem~\cite{HeilmannLieb1972}, the diagonal matching partition
function \(\sum_{M}z^{|M|}\), summed over the matchings \(M\) of a finite graph \(G\),
has only real nonpositive zeros, and Wagner~\cite{Wagner2009} developed a general
method for zero-free regions of Heilmann--Lieb type for spanning
subgraphs with degree constraints.  \Cref{cor:bmatching-short} allows
arbitrary capacities \(b_v\) and gives a neighbourhood of the whole box
\([0,R]^E\), including the interpolation anchor \(\mathbf0\).
\end{remark}

\begin{corollary}[\(b\)-edge covers]
\label{cor:bcover-short}
Fix an integer \(\Deg\ge2\) and
\(0<\lambda_-\le\lambda_+<\infty\).  There is
\(\eps=\eps(\Deg,\lambda_-)>0\) such that, for every finite simple graph
\(G=(V,E)\) of maximum degree at most \(\Deg\), every integer vector
\(\mathbf b=(b_v)_{v\in V}\) with
\(0\le b_v\le\deg_G(v)\), and every \(\mathbf z\in\C^E\),
\[
  Z^{\mathrm{cover}}_{G,\mathbf b}(\mathbf z)
  :=\sum_{\substack{S\subseteq E\\
                    \deg_S(v)\ge b_v\ \forall v}}
       \prod_{e\in S}z_e
  \ne0
\]
whenever \(\mathbf z\in\mathcal U_\eps([\lambda_-,\lambda_+])^E\).
Thus
\(\bigcup_{0<\lambda_-\le\lambda_+<\infty}
  \mathcal U_{\eps(\Deg,\lambda_-)}([\lambda_-,\lambda_+])^E\)
is an open zero-free neighbourhood of \((0,\infty)^E\).
\end{corollary}

\begin{proof}
For \(b'_v:=\deg_G(v)-b_v\), complementation of the selected edges gives
\[
  Z^{\mathrm{cover}}_{G,\mathbf b}(\mathbf z)
  =\left(\prod_{e\in E}z_e\right)
    Z^{\mathrm{match}}_{G,\mathbf b'}
       \bigl((z_e^{-1})_{e\in E}\bigr).
\]
Thus \(\boldsymbol\infty\) is the natural interpolation anchor, corresponding
to \(\mathbf0\) in the reciprocal variables.
Let \(\eps_1\) be the width \(\eps(\Deg,R)\) of
\Cref{cor:bmatching-short} with \(R=\lambda_-^{-1}\); it depends only on
\((\Deg,\lambda_-)\).  Put
\(\eps:=\min\{\lambda_-/2,\ \lambda_-^2\eps_1/2\}\).  If
\(|z_e-x_e|<\eps\) with \(x_e\in[\lambda_-,\lambda_+]\), then
\(|z_e|\ge x_e/2>0\) and
\[
  |z_e^{-1}-x_e^{-1}|=\frac{|z_e-x_e|}{|z_e|\,x_e}
  <\frac{2\eps}{\lambda_-^2}\le\eps_1,
\]
while \(x_e^{-1}\in[\lambda_+^{-1},\lambda_-^{-1}]\subseteq
[0,\lambda_-^{-1}]\).  Hence \(\prod_ez_e\ne0\) and
\((z_e^{-1})_{e\in E}\in\mathcal U_{\eps_1}([0,\lambda_-^{-1}])^E\), so
\Cref{cor:bmatching-short} applies; the width depends neither on \(G\)
nor on \(\lambda_+\).
\end{proof}

The same diagonal specialization in
\Cref{cor:bmatching-short,cor:bcover-short} gives common open univariate
zero-free neighbourhoods of \([0,\infty)\) and \((0,\infty)\), respectively,
depending only on \(\Deg\).

\clearpage
\appendix
\crefalias{section}{appendix}
\numberwithin{equation}{section}
\numberwithin{table}{section}
\renewcommand*{\theHequation}{\theHsection.\arabic{equation}}

\section{Further zero-free regimes for the Potts model}
\label{app:additional-potts}

This appendix records further Potts zero-free regimes obtained by
combining the transfer results of the main text with additional
coupling-independence bounds.  The source of each coupling input is
specified below.

The high-temperature row of \Cref{tab:appendix-extensions} follows from
\Cref{prop:full-interval-hci} and
\Cref{lem:potts-positive-response} of this paper, and the near-Vigoda row
reduces by integrality to \Cref{thm:intro-main} and the Carlson--Vigoda
row.  The remaining coupling-independence inputs, which extend known
hard-colouring couplings to positive activities or establish new ones,
are proved in the companion article~\cite{ShaoShi2026FurtherPotts}.

The companion proofs of these inputs were produced by AI models ---
GPT-5.6 Sol, except the girth-five results, which were proved by GPT-6
Astra.  A Lean formalization, developed with assistance from GPT-6
Astra and Claude Opus 5.5, proves the
coupling-independence and zero-free statements for these regimes~\cite{ShaoShi2026FurtherPotts};
their public entry points retain only the displayed mathematical hypotheses.
It also proves, in the form the proofs use, the results cited in the
three girth rows, \cite[Lemmas~5.13 and~8.7 and Equation~(10)]{CLMM2023}
and \cite[Proposition~2.6(i) and Theorem~2.5]{BencsBerrekkalRegts2025Trees}.
In some places the Lean proof follows a different route from the written
one.  At the twenty near-Vigoda pairs \((q,\Deg)=(11j,6j)\), \(j\le20\),
where the written proof cites \cite[Theorem~20]{CFFGZZ2024}, it extends
the Carlson--Vigoda contraction to \(q\ge11\Deg/6\) for every
\(\Deg\ge6\); and it treats the cases \(\Deg=q+2\) of the BBR row
(Bencs, Berrekkal, and Regts~\cite{BencsBerrekkalRegts2025Trees})
directly.  More generally, the formalization proves every numbered
theorem, lemma, proposition, and corollary of this article and
of~\cite{ShaoShi2026FurtherPotts}, some in an equivalent form, among them
\Cref{thm:intro-main,thm:potts-transfer}, the case \(q=11\Deg/6\)
included, and the colour-field and Holant results of
\Cref{sec:lee-yang,sec:holant-consequence}, all without external input.
Only citation-level claims and the algorithmic remarks of
\Cref{sec:intro} are not formalized.

\Cref{tab:appendix-extensions} gives the additional zero-free regions.
The coupling input of each row originates as follows.
\begin{itemize}[leftmargin=1.5em,itemsep=1pt]
  \item Near-Vigoda: by
        integrality~\cite[Lemma~4.3]{ShaoShi2026FurtherPotts}, the union
        of \Cref{thm:intro-main} (with \cite[Theorem~20]{CFFGZZ2024} in
        the equality case \(q=11\Deg/6\), that is,
        \((q,\Deg)=(11j,6j)\) with \(j\ge1\)) and the Carlson--Vigoda
        row.  The constant
        \(1/84000\) is the value of \(\eps_0\) fixed in the proof
        of~\cite[Theorem~5.1]{CDMPP2019}.
  \item Carlson--Vigoda: the contraction of~\cite{CarlsonVigoda2024}
        (at \(x=0\), \cite[Theorem~20]{CFFGZZ2024}), extended to
        \(x\in[0,1]\)~\cite[Theorem~5.1]{ShaoShi2026FurtherPotts}.
  \item Large-girth: the mechanism of~\cite{CLMM2023} (at \(x=0\), as in
        \Cref{thm:lee-yang}), extended to
        \(x\in[0,1]\)~\cite[Theorem~6.1]{ShaoShi2026FurtherPotts}.
  \item High-temperature: \Cref{prop:full-interval-hci}.
  \item BBR interval: the tree contraction
        of~\cite{BencsBerrekkalRegts2025Trees}, strengthened to tree
        total-influence decay and transferred to large-girth graphs by
        the mechanism of~\cite{CLMM2023};
        see~\cite[Theorem~7.6]{ShaoShi2026FurtherPotts}.
  \item Edge-Potts: a slot-lift adaptation of the list-edge-colouring
        coupling of~\cite{CWZZ2025} (whose \(x=0\) bound enters
        \Cref{cor:edge-lee-yang}), with coupling-independence constant
        \(\Deg-1\) on all of \([0,1]\), including
        \(x=0\)~\cite[Theorem~8.1]{ShaoShi2026FurtherPotts}.
  \item Girth-five: a covariance recursion built on the all-pinning
        spectral gap of~\cite[Theorem~9.7]{ShaoShi2026FurtherPotts},
        converted by the sphere-to-coupling
        lemma~\cite[Lemma~5.13]{CLMM2023} into the coupling
        bound~\cite[Theorem~9.4]{ShaoShi2026FurtherPotts}.  This setting
        parallels the spectral local-to-global framework
        of~\cite{ChenLiu2026LocalGlobal}, from which no input is taken.
\end{itemize}

\begin{table}[tbp]
\centering
\small
\setlength{\tabcolsep}{7pt}
\renewcommand{\arraystretch}{1.22}
\begin{tabular}{@{}lll@{}}
\toprule
Regime & Parameters & Zero-free near \\
\midrule
near-Vigoda
  & \(\Deg\ge2\), \(q\ge(11/6-1/84000)\Deg\)
  & \([0,1]\) \\
Carlson--Vigoda
  & \(\Deg\ge125\), \(q\ge1.809\Deg\)
  & \([0,1]\) \\
large-girth
  & \(\Deg\ge3\), \(q\ge\Deg+3\); \(\operatorname{girth}(G^\tau)\ge g_*(q,\Deg)\)
  & \([0,1]\) \\
high-temperature
  & \(\Deg\ge2\), \(q>(11/6)(1-x_*)\Deg\), \(x_*\in(0,1]\)
  & \([x_*,1]\) \\
BBR interval
  & \(q\ge3\), \(\Deg/q\ge(e-1/2)/(e-1)\); \(\operatorname{girth}(G^\tau)\ge g_{\rm BBR}(q,\Deg)\)
  & \([x_0,1]\) \\
edge-Potts
  & \(\Deg\ge2\), \(q\ge3\Deg\)
  & \([0,1]\) \\
girth-five
  & \(q\ge(1+\delta)\Deg\), \(0<\delta\le1\); \(\Deg\ge\Deg_5(\delta)\), \(\operatorname{girth}(G^\tau)\ge5\)
  & \([0,1]\) \\
\bottomrule
\end{tabular}
\smallskip

\parbox{0.97\textwidth}{\raggedright\footnotesize
In every row the conclusion is \(\nZpin{G}{\tau}(z)\ne0\) for all \(z\)
in an \(\eps\)-neighbourhood of the stated interval, with \(\eps>0\)
depending only on the listed parameters, uniformly over graphs and
pinnings.  The unnormalized conclusions are specified in
\Cref{rem:additional-potts-transfer}.
In the edge-Potts row the polynomial is \(\nZpin{L(G)}{\tau}\), with
\(\tau\) an edge-colour pinning.  The thresholds \(g_*(q,\Deg)\),
\(g_{\rm BBR}(q,\Deg)\), and \(\Deg_5(\delta)\) are those
of~\cite[Theorems~6.1, 7.6, and~9.4]{ShaoShi2026FurtherPotts}; BBR refers
to Bencs, Berrekkal, and Regts~\cite{BencsBerrekkalRegts2025Trees}.  In
the BBR row,
\(x_0:=1-\frac q\Deg\bigl(1-\frac k\Deg\bigr)^{\!2}
\frac{\Deg-k}{\Deg-k/2}\) with \(k:=e\frac{\Deg-q/2}{\Deg-q}\) (and
\(x_0:=3/4\) if \((q,\Deg)=(3,4)\)).  The girth-five row assumes no
bounds on local short-cycle counts.}
\caption{Further Potts zero-free regimes obtained from coupling
independence.}
\label{tab:appendix-extensions}
\end{table}

\begin{remark}[Transfer and the hard endpoint]
\label{rem:additional-potts-transfer}
The full-interval Potts rows use \Cref{thm:potts-transfer}
(on \(L(G)\), with degree bound \(2\Deg-2\), for edge-Potts).
The high-temperature and BBR rows use the positive-activity induction in
\Cref{lem:potts-positive-response}.  In the rows with a condition
\(\operatorname{girth}(G^\tau)\ge g\) (\(g=g_*\), \(g_{\rm BBR}\), or \(5\)),
these results are applied to \(\Gdegg{g}\) after
replacing \((G,\tau)\) by \((\widetilde G,\widetilde\tau)\): starting from
\(G^\tau\), for every edge \(up\) with \(u\in V^\tau\) and
\(p\in\Lambda\), attach to \(u\) a new leaf pinned to \(\tau_p\).  This
preserves \(\nZpin{G}{\tau}\) and the degree in \(G\) of every free vertex,
the new leaves have degree one, and
\(\operatorname{girth}(\widetilde G)=\operatorname{girth}(G^\tau)\); hence
\(\widetilde G\in\Gdegg{g}\).  For rows with interval \([0,1]\), the
unnormalized pinned polynomial has no zeros in the stated neighbourhood
except, when \(m_G(\tau)\ge1\), a zero of multiplicity \(m_G(\tau)\) at
\(0\) (\(m_{L(G)}(\tau)\) for edge-Potts).
For the two positive-interval rows, choose \(\eps\) smaller than the
left endpoint; the unnormalized polynomial is then nonzero throughout
the same neighbourhood.
\end{remark}

\bibliographystyle{alpha}
\bibliography{main}

\end{document}